\documentclass[letterpaper]{article}

\makeatletter
\newlength\titlebox 
\twocolumn
\long\gdef\affiliations #1{\def \affiliations_{#1}}%

\def\maketitle{%
  \par%
  \begingroup
    \def\thefootnote{\fnsymbol{footnote}}%
    \twocolumn[\@maketitle]\@thanks%
  \endgroup%
  \setcounter{footnote}{0}%
  \let\maketitle\relax%
  \let\@maketitle\relax%
  \gdef\@thanks{}%
  \gdef\@author{}%
  \gdef\@title{}%
  \let\thanks\relax%
}%

\def\@maketitle{%
  \newcounter{eqfn}\setcounter{eqfn}{0}%
  \newsavebox{\titlearea}%
  \sbox{\titlearea}{%
    \let\footnote\relax\let\thanks\relax%
    \setcounter{footnote}{0}%
    \vbox{%
      \hsize\textwidth%
      \linewidth\hsize%
      \vskip 0.625in minus 0.125in%
      \centering%
      {\LARGE\bf \@title \par}%
      \vskip 0.1in plus 0.5fil minus 0.05in%
      {\Large{\textbf{\@author\ifhmode\\\fi}}}%
      \vskip .2em plus 0.25fil%
      {\normalsize \affiliations_\ifhmode\\\fi}%
      \vskip 1em plus 2fil%
    }%
  }%
  \newlength\actualheight%
  \settoheight{\actualheight}{\usebox{\titlearea}}%
  \ifdim\actualheight>\titlebox%
    \setlength{\titlebox}{\actualheight}%
  \fi%
  \vbox to \titlebox{%
    \hsize\textwidth%
    \linewidth\hsize%
    \vskip 0.625in minus 0.125in%
    \centering%
    {\LARGE\bf \@title \par}%
    \vskip 0.1in plus 0.5fil minus 0.05in%
    {\Large{\textbf{\@author\ifhmode\\\fi}}}%
    \vskip .2em plus 0.25fil%
    {\normalsize \affiliations_\ifhmode\\\fi}%
    \vskip 1em plus 2fil%
  }%
}%

\renewenvironment{abstract}{%
  \centerline{\bf Abstract}%
  \vspace{0.5ex}%
  \setlength{\leftmargini}{10pt}%
  \begin{quote}\small%
}{%
  \end{quote}%
  \vskip 1ex%
}%

\def\section{\@startsection {section}{1}{\z@}{-2.0ex plus
-0.5ex minus -.2ex}{3pt plus 2pt minus 1pt}{\Large\bf\centering}}
\def\subsection{\@startsection{subsection}{2}{\z@}{-2.0ex plus
-0.5ex minus -.2ex}{3pt plus 2pt minus 1pt}{\large\bf\raggedright}}
\def\subsubsection{\@startsection{subparagraph}{3}{\z@}{-2pt plus
-1pt minus -1pt}{-1em}{\normalsize\bf}}
\renewcommand\paragraph{\@startsection{paragraph}{4}{\z@}{-6pt plus -2pt minus -1pt}{-1em}{\normalsize\bf}}%
\skip\footins 9pt plus 4pt minus 2pt
\def\footnoterule{\kern-3pt \hrule width 5pc \kern 2.6pt}
\labelwidth\leftmargini\advance\labelwidth-\labelsep \labelsep 5pt

\belowdisplayskip \abovedisplayskip
\def\normalsize{\@setfontsize\normalsize\@xpt{11}}
\def\small{\@setfontsize\small\@ixpt{10}}
\def\footnotesize{\@setfontsize\footnotesize\@ixpt{10}}
\def\scriptsize{\@setfontsize\scriptsize\@viipt{10}}
\def\tiny{\@setfontsize\tiny\@vipt{7}}
\def\large{\@setfontsize\large\@xipt{12}}
\def\Large{\@setfontsize\Large\@xiipt{14}}
\def\LARGE{\@setfontsize\LARGE\@xivpt{16}}
\def\huge{\@setfontsize\huge\@xviipt{20}}
\def\Huge{\@setfontsize\Huge\@xxpt{23}}

\makeatother

\usepackage{times}  
\usepackage{helvet} 
\usepackage{courier} 
\usepackage[hyphens]{url}  
\usepackage{graphicx} 
\usepackage{natbib} 
\usepackage{caption}
\usepackage{enumitem}

\usepackage{newfloat}
\usepackage{listings}
\DeclareCaptionStyle{ruled}{
labelfont=normalfont,
labelsep=colon,
strut=off
} 
\usepackage{graphicx}
\usepackage{amsmath,amssymb,amsthm}
\usepackage{graphicx}
\usepackage{subcaption}
\usepackage{comment}
\usepackage{breqn}
\usepackage[ruled, linesnumbered]{algorithm2e}
\usepackage{booktabs, longtable, multicol, adjustbox}
\usepackage{tabularx}
\usepackage{mathtools}
\usepackage{color}
\usepackage{soul, xcolor}
\usepackage{tikz}
\usetikzlibrary{arrows.meta, positioning}

\usepackage[hidelinks]{hyperref}
\usepackage{cleveref}
\usepackage{hypernat}

\newcommand{\Amax}{A^{\max}}
\newcommand{\Apos}{A^{+}}

\newtheorem{theorem}{Theorem}[section]
\newtheorem{lemma}[theorem]{Lemma}
\newtheorem{proposition}[theorem]{Proposition}
\newtheorem{corollary}[theorem]{Corollary}

\theoremstyle{definition}
\newtheorem{definition}[theorem]{Definition}

\newtheorem{example}[theorem]{Example}

\DeclareMathOperator{\supp}{supp}
\DeclareMathOperator*{\argmax}{\arg \max}

\title{Efficiently Computing Equilibria in Budget-Aggregation Games}
\author{
Patrick Becker\textsuperscript{\rm 1},
Alexander Fries\textsuperscript{\rm 1},
Matthias Greger\textsuperscript{\rm 2},
Erel Segal-Halevi\textsuperscript{\rm 3}
}
\affiliations{
\textsuperscript{\rm 1}Technical University of Munich, Germany\\
\textsuperscript{\rm 2}Paris Dauphine - PSL, France\\
\textsuperscript{\rm 3}Ariel University, Israel
}

\begin{document}
\maketitle
\begin{abstract}    
    Budget aggregation deals with the social choice problem of distributing an exogenously given budget among a set of public projects, given agents' preferences. Taking a game-theoretic perspective, we study \emph{budget-aggregation games} where each agent has virtual decision power over some fraction of the budget. We investigate the structure and show efficient computability of Nash equilibria for various common preference models in this setting. In particular, we show that equilibria for Leontief utilities can be found in polynomial time, solving an open problem from \citet{BGSS23b}, and give an explicit polynomial-time algorithm for computing equilibria for $\ell_1$ preferences. 
\end{abstract}

\section{Introduction}
\label{sec:introduction}
\emph{Participatory budgeting} \citep{Caba04a,AzSh20a} is a democratic process in which citizens are allowed to vote on how the funds of their city or state should be distributed. 
It is justified both by ethical arguments (people have a fundamental right to decide how their money is used) and by practical arguments (letting people decide on the budget strengthens cohesion in the community).

This approach can be taken to the extreme: instead of just letting citizens vote on which projects they approve (as is commonly done today), one can give each citizen their own (virtual) share of the budget (e.g. $1/n$ of the total budget, in case there are $n$ citizens with equal entitlements), and allow them to decide how to allocate that share among the potential projects. 
This can be an attractive approach, as it gives full power to the citizens and ensures some degree of fairness by splitting this power equally among them. Furthermore, each citizen is able to directly observe her influence on the outcome by means of the distribution of her individual budget. 

However, this approach raises severe strategic problems. For example, if a citizen thinks that projects $x$ and $y$ should each receive $0.5$ of the total budget, but later observed that most other citizens allocated their shares to $x$, then she would have preferred to spend her entire budget on $y$, leaving her regretful since her budget was spent sub-optimally in hindsight. 
This bad user experience can be avoided if we have the possibility to coordinate individual budgets by advising each citizen in advance how to distribute their own budget such that for each citizen, the advised distribution is optimal, given that all other citizens follow the provided advice.

In other words, we would like to compute a \emph{Nash equilibrium (NE)} of the game in which the players are citizens, and the set of strategies of each player is the set of possible distributions of their share of the budget.
We call this game the \emph{budget-aggregation game}.
An NE of the budget-aggregation game hits a sweet spot between giving complete autonomy to the citizens to distribute their shares of the budget and ensuring that they are indeed satisfied with their own choice.
More specifically, every agent receives at least their maximin share, a property studied in fair division \citep[see, e.g.,][]{BCE+15a}. In words, an agent receives the highest possible utility they can secure with their share of the budget under the least favorable allocation of the remaining shares.

\paragraph{Our contributions}
We propose a new class of normal-form games, called budget-aggregation games, which model (equal-resource) public funding situations. We position these games within the broader framework of normal-form games and survey existing results on the computation of NEs.
On top of that, our paper analyzes NEs for the most common utility functions that have been considered in participatory budgeting and related fields.
We prove theorems on the structure of NEs and develop algorithms to efficiently compute them. 
We focus on the following utility models:

\begin{itemize}[noitemsep,topsep=0pt]
\item For \textit{linear utilities}, we show that an NE can be computed by solving a linear program (\Cref{sec:linear}).
\item For \textit{Leontief utilities}, we show that an NE can be computed in polynomial time using a variant of the ellipsoid method, solving an open problem from \citet{BGSS23b} (\Cref{sec:Leontief}).
\item For \textit{binary symmetric separable utilities}, we explore strategic similarities to Leontief utilities to show that an NE can be found in polynomial time.
(\Cref{sec:symsep}).
\item For \textit{$\ell_1$ disutilities}, we again present an intuitive polynomial-time algorithm to compute an NE (\Cref{sec:l1_preferences}).
\end{itemize}
Finally, we extend the basic model by allowing agents to have different weights.
We show that, for most of the above utility models (namely linear, Leontief and $\ell_1$), an NE can still be computed in pseudo-polynomial time (\Cref{sec:weighted}).

\section{Related Work}
About $75$ years ago, two groundbreaking papers by John Nash \citep{Nash50a,Nash51a} proved the existence of an NE in general normal-form games. Apart from its intrinsic value, the proof technique of applying fixed-point theorems, especially Kakutani's fixed-point theorem (which Nash attributed to David Gale), paved the way for other equilibrium existence theorems like \citeauthor{Debr52a}'s social equilibrium \citep{Debr52a} or general equilibrium in the Arrow-Debreu model \citep{ArDe54a,McKe59a}.

An important generalization of Nash's theorem involves the notion of a \emph{concave game} \citep{rosen1965existence}. A concave game is defined by, for each player $i$, a convex set 
$(S_i)_{i\in N}$
that represents the possible strategies for $i$, and a utility function 
$u_i$
that maps each strategy profile to a real number; $u_i$ can be any function that is continuous in all strategies, and concave in the strategy of $i$.%
\footnote{The definition in \citet{rosen1965existence} is even more general and allows \emph{coupled constraints}, that is, the set of strategies available to each agent may depend on the choices made by other agents. 
We do not need this generalization here.
}
Nash's game is a special case in which every $S_i$ is a standard simplex in some Euclidean space (representing the set of all lotteries over pure actions), and every $u_i$ is a linear combination of products of $n$ terms.
Our budget-aggregation game is a special case in which every $S_i$ is (still) a simplex, but the utility functions are substantially different, due to the interpretation as a distribution of exogenously given budget rather than a lottery.

\subsection{Computing an equilibrium is hard}

Like most fixed-point-based proofs, the ones by \citet{Nash50a,Nash51a} and \citet{rosen1965existence}
do not yield efficient algorithms for finding an equilibrium. 
Moreover, there are 3-player games with integer payoffs with a unique (mixed) NE in which all probabilities are irrational numbers \citep{Nash51a,bilo2014complexity}, and 4-player games with a unique NE in which all probabilities are \emph{irradical} numbers (cannot be expressed using integer roots) \citep{orzech2025nash}. This shows that no finite algorithm can compute an exact NE in general games.

Even computing an $\epsilon$-approximate NE probably cannot be done in polynomial time, as it is PPAD-hard even for two players. This was first proven for $\epsilon$ that shrinks exponentially in the number of actions \citep{chen2006settling}, then for $\epsilon$ that shrinks polynomially \citep{chen2006computing}, and later even for some constant $\epsilon>0$ \citep{rubinstein2015inapproximability}. 
Computing an equilibrium is PPAD-complete, also for some other classes of concave games
\citep{papadimitriou2023computational}.
Many natural decision problems regarding the existence of a NE with certain properties are NP-hard for two-player games, and $\exists \mathbb{R}$-hard for multi-player games \citep{bilo2014complexity}.

However, there are other classes of games for which an NE can be computed efficiently. 
We survey some of them below.

\subsection{Nash games}
For the original game studied by Nash (where the strategy set of each agent is the set of lotteries over a finite set of actions),
many practical algorithms for computing an NE were designed, such as the ones by \citet{lemke1964equilibrium} and \citet{porter2008simple}, but without a polynomial time guarantee.

Polynomial-time algorithms for $\epsilon$-approximate NE (where no agent can increase her utility by more than $\epsilon$ by deviating, where the utilities are normalized to $[0,1]$) are known for $\epsilon\approx 0.34$ for two players \citep{tsaknakis2007optimization}, $\epsilon\approx 0.5$ for $n$-player polymatrix games, and $\epsilon > 0.6$ for general $n$-player games with $n\geq 3$ \citep{deligkas2017computing}.
~
Some special classes of games in which an exact NE can be computed in polynomial time are:
\begin{itemize}[noitemsep,topsep=0pt]
\item Two-player zero-sum games \citep{vonneumann1928theorie};
\item Anonymous games in which each player has a fixed small number of strategies and Lipschitz-continuous utility functions \citep{daskalakis2007computing,daskalakis2015approximate,cheng2017playing};
\item Graphical games in which the graph (representing the influence between players) is a tree \citep{kearns2001graphical} or a general graph with degree 2 \citep{elkind2006nash}; 
\item Win-lose games with at most 2 winning positions per action of each player \citep{codenotti2006efficient};
\item Two-player games in which the payoff matrices have a bounded rank \citep{lipton2003playing,adsul2011rank}.
\end{itemize}

\subsection{Aggregative games}
An \emph{aggregative game} is a game in which the utility of each player $i$ depends only on the strategy of $i$, and on some aggregate function of all strategies. Our game corresponds to the special case in which that aggregate function is a sum.

Aggregative games were initially studied for one-dimensional strategy sets (each agent chooses a number) and one-dimensional aggregators. 
Under certain conditions, such games admit a \emph{potential function},%
\footnote{Roughly, a potential function is a function of the strategy profile, whose value increases when any player makes an improving move. A weaker definition only requires the potential to increase when any player makes a best response.}
and thus always have an NE
\citep{kukushkin1994fixed, dubey2006strategic, cornes2012fully, jensen2018aggregative}. 
There are also efficient algorithms for computing an NE when players have discrete strategy sets \citep{kearns2002efficient}.

For aggregative games with continuous strategy sets, \citet{babichenko2018fast} 
proves that, for every $\epsilon>0$, there exists a best-reply dynamics that reaches an $O(\lambda \epsilon)$-approximate NE in time (number of steps) $O(n \log n / \epsilon^2)$, where $\lambda$ is the Lipschitz constant of the agents' utility functions.
This result is applicable to a special case of our budget-aggregation game in which there are only $m=2$ issues, as in that case, the strategy of each player is completely determined by the amount he puts on the first issue (the rest of his share automatically goes to the second issue).
As far as we know, it is open whether an \emph{exact} NE for $m=2$ can be computed in polynomial time for this general class of games.

\citet{cummings2015privacy} study aggregative games with a multi-dimensional aggregator (as in our game), but the players' strategy sets are finite (unlike our game). For this class, they present an algorithm that computes an NE in time polynomial in the number of players and actions, but exponential in the number of dimensions of the aggregator.

When both players' strategy sets and the aggregator are multi-dimensional, one can still prove that, under the required assumptions, they have a potential function, and therefore have an NE \citep{jensen2010aggregative}. However, computing the NE becomes harder. 
The literature focuses on numeric, distributed algorithms that converge asymptotically to an NE \citep[e.g.,][]{parise2020distributed,wang2024differentially},
see \citet{li2025survey} for a recent survey.
We could not find an algorithm with a polynomial run-time guarantee.

Against this background, our contribution is identifying a natural and practical class of aggregative games, for which an exact NE can be computed efficiently.
In the following, different types of aggregative games are covered.

\subsubsection{Congestion games}
These games are a subclass of aggregative games with multidimensional aggregators, where each coordinate of the aggregator represents the congestion on one of the resources.
Usually, the players' strategy sets are \emph{discrete}; each player can choose a subset of resources to use (lotteries are not allowed). Equilibria without lotteries are called \emph{pure} Nash equilibria (PNE).

Even though Nash's original existence proof is not valid in general for discrete strategy sets, congestion games (CGs) always have a potential function, so every sequence of individual improvements must terminate in a PNE.
However, in general, this sequence might be exponentially long.
Specifically, computing a PNE in general CGs is PLS-complete, which means that it probably cannot be done by a polynomial-time algorithm \citep{fabrikant2004complexity}.%
\footnote{
PLS is similar to PPAD in that both contain problems in which a solution is guaranteed to exist --- both are subsets of TFNP. It is believed that none of them contains the other, and that none of them is contained in P.}

Some special classes of CGs in which a PNE can be computed or approximated in polynomial time are:
\begin{itemize}[noitemsep,topsep=0pt]
\item Symmetric network CGs:
all players have the same set of possible strategies, and this set is the set of paths between a source node and a target node in a given network \citep{fabrikant2004complexity}.
\item Nonatomic CGs: CGs in which there is a continuum of players, where each player's effect on the outcome is negligible (in contrast to our setting, where each player has a substantial effect on the outcome).
Finding a PNE in a nonatomic CG can be reduced to solving a convex optimization problem, which can be done in weakly polynomial time.
In the special case of a \textit{network} nonatomic CG with Lipschitz-continuous utility functions, there is an FPTAS (an algorithm computing an $\epsilon$-approximate PNE in time strongly-polynomial in the problem size and $1/\epsilon$). \citep{fabrikant2004complexity}.
\item CGs in which the strategy space of each player consists of the bases of a matroid over the set of resources \citep{ackermann2008impact}.
\item For symmetric (non-network) CGs, an FPTAS based on simple greedy dynamics \citep{chien2011convergence}.
\item For general CGs, when the delay functions are polynomials of bounded degree, there are dynamics based on clever sequences of greedy steps that obtain a constant-factor approximate PNE \citep{caragiannis2015approximate}.
\end{itemize}

The CG class most closely related to our budget-aggregation game is the class of \emph{splittable (atomic) CGs}  \citep{orda2002competitive,roughgarden2015local}, in which each agent controls a divisible load and can split it arbitrarily among different paths in a network. 
The players' utilities are decreasing functions of the total load allocated to edges that they use: $u_i(\delta) = - \sum_{x\in A} \delta_{i,x}\cdot c_{i,x}(\delta_x)$, where $c_{i,x}$ is the congestion function of edge $x$ for player $i$. An important special case is when the congestion functions are affine, $c_{i,x}(\delta_x) = a_{i,x}\cdot \delta_x + b_{i,x}$. 
\citet{klimm2025complexity} prove that, even with affine congestion functions,  computing a PNE is PPAD-complete. Still, polynomial-time algorithms are known for some special cases.  If the congestion functions are affine and player-independent, an approximate PNE can be computed by convex programming for general graphs \citep{cominetti2009impact}; for specific classes of graphs called "well-designed", an exact PNE can be computed by a combinatorial algorithm \citep{huang2013collusion}. When strategy sets are singletons --- the strategy set of each player is a set of edges, rather than a set of paths from a source to a target --- there is a polynomial-time algorithm for affine player-specific congestion functions \citep{harks2022equilibrium}; for convex player-independent congestion functions, there are algorithms that are polynomial when either the number of players or the number of edges is fixed \citep{bhaskar2018equilibrium}.

Note that the utility functions are different than our linear utility functions in that the (negative) utility that a player derives from an edge is multiplied by the player's own contribution to that edge. In other words, the dependence of $u_i(\delta)$ on $\delta_{i,x}$ is quadratic, as $\delta_{i,x}$ both multiplies $c_{i,x}$ and contributes to $\delta_x$.

\subsubsection{Budget-aggregation games with approval preferences}

The need to study NEs in budget-aggregation games is also reflected by the fact that many fairness-motivated solution concepts in participatory budgeting and multiwinner voting \citep[see, e.g.,][]{LaSk23a} are related to NEs, especially when agents have approval preferences, i.e., linear utilities with binary valuations, also called dichotomous utilities \citep{BMS05a}. For the uncapped setting of budget-aggregation, when projects can receive arbitrary amounts of the budget, \citet{BBP+19a} coined the term ``decomposability'' for distributions of the overall budget that admit a decomposition into distributions of the individual budgets where each agent spends her budget optimally, i.e., only on approved projects. 
The standard model in participatory budgeting assumes that projects have fixed costs. For such problems, common rules like the \textit{method of equal shares} \citep{PeSk20a,PPS21a} or \textit{Phragm\'en's rule} \citep{Phra99a,Jans16a} rely on the same idea that agents successively allocate their budgets to fund approved projects.
Most of these works take a ``mechanism design perspective'' where agents report their preferences, which are then aggregated into a collective outcome, instead of individual budget distributions.
A notable exception is \emph{donor coordination} \citep[see, e.g.,][]{FPP+19a} where the budget is initially owned by the agents. \citet{BGSS23b,BGSS23a} proved that, with Leontief utilities, the unique equilibrium distribution is the unique distribution that maximizes the product of agents' utilities where agents are weighted by their contributions.
For the special case of multiwinner approval voting (where all projects have the same costs and correspond to candidates), \citet{HKMS24a} recently defined a \emph{budgeting game}, which is very similar to our budget-aggregation game. In their game, as in ours, each player has a fixed budget, and can choose a distribution of this budget among $m$ issues.
But in contrast to our game, their issues are \emph{discrete}; each issue is \emph{active} if its total allotment is at least $1$ and \emph{inactive} otherwise. The utility of a player is the number of active issues that the player approves. 
They prove that their budgeting game always has an NE, but do not elaborate on its computational aspects.
Moreover, returning NEs already ensures some degree of fairness as they satisfy \textit{individual fair share} \citep{BMS05a,ABM20a} for arbitrary preference models (see \Cref{prop:NashIFS}).

\subsection{Computing market equilibria}
Another equilibrium concept which has been extensively studied from a computational perspective is \emph{competitive equilibrium (CE)} in markets for private commodities, such as a \emph{Fisher market} (one seller and many buyers) or the more general \emph{exchange market} (every agent can both sell and buy; see e.g. \citet{vazirani2007combinatorial} for the definitions). 

As in the case of NE, the existence of CE in very general settings has been proved using fixed-point theorems \citep{Debr52a,ArDe54a,McKe59a};
this lead to the development of approximation algorithms without polynomial-time guarantee \citep{scarf1967computation,merrill1972applications}.

Later, it has been proved that computing an approximate CE is PPAD-hard in general \citep{Papa94b}, and even for some special cases such as exchange markets with Leontief utilities \citep{devanur2008marketdual}, exchange markets with separable piecewise-linear concave utilities \citep{chen2009settling}, Fisher markets with such utilities, \citep{chen2009spending} and exchange markets with negative linear utilities \citep{chaudhury2020dividing}.
Computing an exact CE is FIXP-hard in general \citep{etessami2010complexity}, and even for the special case of Leontief utilities \citep{garg2017settling}.

However, CE can be computed in polynomial time for some special cases (some of these cases require either the number of goods $m$ or the number of agents $n$ to be fixed):
\begin{itemize}[noitemsep,topsep=0pt]
\item Exchange markets with Cobb-Douglas utilities \citep{curtis1985finite,curtis2009finite};
\item Exchange markets with linear utilities when $m$ is fixed \citep{deng2002complexity};
\item Fisher markets with Leontief utilities \citep{codenotti2004efficient};
\item Exchange markets with constant elasticity of substitution (CES), with elasticity at least $1/2$ \citep{codenotti2005market};
\item Exchange markets with weak gross substitute utilities \citep{codenotti2005polynomial} --- a generalization of linear, Cobb-Douglas and CES;
\item Fisher markets with logarithmic utilities when either $n$ or $m$ is fixed \citep{chen2004fisher};
\item Exchange markets with piecewise-linear concave utilities when $m$ is fixed \citep{devanur2008marketfixed};
\item Exchange markets with Leontief utilities when $n$ is fixed \citep{garg2017settling};
\item Fisher markets with linear utilities \citep{devanur2008marketdual,orlin2010improved};
\item Fisher markets with linear utilities that are negative \citep{branzei2019algorithms} or even a mixture of positive and negative \citep{garg2020computing}, when either $m$ or $n$ is fixed;
\item Fisher markets with separable piecewise-linear concave negative functions \citep{chaudhury2021competitive}.
\end{itemize}
Most related to our work is an algorithm developed by \citet{Jain07a} for exchange markets with linear utilities; we use some of its theorems in our proofs for Leontief utilities.
He first developed a convex program that characterizes all CEs, and then showed how it can be solved efficiently using a variant of the ellipsoid method. Soon afterwards,
\citet{ye2008path} showed how to solve this program using interior-point methods, which are more efficient in practice
(see also \citet{codenotti2004computation} for a survey of similar methods).

\section{Model}
\label{sec:model}
Let $N$ be a set of $n$ agents and $A$ be a set of $m$ (pure) alternatives. The outcome space $\Delta(A)$ is formed by all distributions $\delta \in [0,1]^m$ with $\sum_{x \in A}\delta_x=1$. 
For a distribution $\delta$, we denote its support by $\supp(\delta) := \{x\in A: \delta_x>0\}$.

In a budget-aggregation game, the budget (w.l.o.g., normalized to $1$) is equally divided among the agents, and each of them is allowed to distribute their share among the alternatives.
The set of strategies available to each player is thus $S_i=\{\delta_i \in [0,1]^m \text{ such that } \sum_{x \in A}\delta_{i,x}=1/n\}$, which is an $(m-1)$-dimensional simplex in $\mathbb{R}^m$.

\begin{definition}
\label[definition]{def:game}
A budget-aggregation game is defined as the tuple $(N, A,(u_i)_{i \in N})$ consisting of
\begin{itemize}[noitemsep,topsep=0pt]
\item a set of agents $N$ with $n=|N|$,
\item a set of alternatives (issues) $A$,
\item and $n$ individual utility functions $u_i \colon \Delta(A) \to \mathbb{R}$ for all $i\in N$.
\end{itemize}
\end{definition}
The strategy of agent $i$, denoted $\delta_i=(\delta_{i,x})_{x \in A}$, implies that agent $i$ allocates $\delta_{i,x}$ to alternative $x$; these strategies are called \emph{individual distributions}.
We denote the \emph{overall distribution} by $\delta := \sum_{i \in N}\delta_i$.
Note that agents only care about the overall distribution of the budget, which is why the domain of $u_i$ is $\Delta(A)$ instead of $S_1\times\cdots\times S_n$. 
Furthermore, let $S_{-i}:=S_1 \times \dots \times S_{i-1}\times S_{i+1}\times \dots \times S_n$ be the space of reduced strategy profiles where we exclude agent $i$.

\begin{definition}
\label[definition]{def:Nasheq}
A strategy profile $(\delta^*_i)_{i \in N}$ is called a \emph{Nash equilibrium} if for all $i \in N$,
\begin{align*}
u_i(\delta^*)=\max_{\delta_i \in S_i} u_i\left(\sum_{j \in N \setminus i}\delta^*_j+\delta_i\right)\text{.}
\end{align*}

The sum $\delta^* = \sum_{i \in N} \delta^*_i$ is called an \emph{equilibrium distribution}.
\end{definition}

Interestingly, NEs already guarantee some degree of fairness, as they satisfy a natural fairness notion called \emph{individual fair share}.
\begin{definition}
\label[definition]{def:IFS}
A strategy profile $(\delta_i)_{i \in N}$ satisfies \emph{individual fair share} if for all $i \in N$,
\begin{align*}
u_i\left(\sum_{j \in N}\delta_j\right) \geq \max_{\delta'_i \in S_i}\min_{(\delta'_j)_{j \neq i}\in S_{-i}}u_i\left(\sum_{j \in N}\delta'_j\right)\text{.}
\end{align*}
\end{definition}
In words, individual fair share ensures that each agent obtains at least as much utility as she could guarantee on her own, by allocating her budget independently of the others.

\begin{proposition}\label{prop:NashIFS}
Every Nash equilibrium $(\delta^*_i)_{i \in N}$ satisfies individual fair share.
\end{proposition}
\begin{proof}
For every NE $(\delta^*_i)_{i \in N}$,
\begin{align*}
u_i(\delta^*)&=\max_{\delta_i \in S_i} u_i\left(\sum_{j \in N \setminus i}\delta^*_j+\delta_i\right) \\
&\geq \max_{\delta_i \in S_i}\min_{(\delta_j)_{j \neq i}\in S_{-i}}u_i\left(\sum_{j \in N\setminus i}\delta_j+\delta_i\right)\text{.}  
\end{align*}
Thus, $(\delta^*_i)_{i \in N}$ satisfies individual fair share. 
\end{proof}

We are interested in understanding the structure and computing NEs for various utility models.
In the following sections, we will consider linear utilities (\Cref{sec:linear}), Leontief utilities (\Cref{sec:Leontief}), binary symmetric separable utilities (\Cref{sec:symsep}), and $\ell_1$ disutilities (\Cref{sec:l1_preferences}).

\section{Linear utilities}\label{sec:linear}
\emph{Linear utilities}, also known as \emph{additive} or \emph{von Neumann-Morgenstern utility functions}, are the most common extension of valuations for a set of pure alternatives $A$ to preferences over $\Delta(A)$. When agents' preferences satisfy four seemingly mild conditions, \citet{vNM44a} showed that they can be represented by linear utility functions. Formally, $i$ has linear utilities if for all $\delta\in \Delta(A)$,
\begin{align*}
u_i(\delta)=\sum_{x \in A}v_{i,x}\delta_x.
\end{align*}
with nonnegative valuations $(v_{i,x})_{x \in A}$ for alternatives.\footnote{Agents do not have different valuations for individual distributions $\delta_{i,x}$, $\delta_{j,x}$ to the same alternative $x$ as they only care about the total amount $\delta_x$ allocated to $x$.}

From an economic perspective, alternatives admit an interpretation as substitutes, as agent $i$ is indifferent to how the budget is divided between two alternatives $x,y$ with $v_{i,x}=v_{i,y}$.
Due to this property, the set of NEs admits an easily understandable structure.
Let $\Amax_i:= \argmax_{x \in A} v_{i,x}$ be the set of alternatives highest valued by agent $i$. 

\begin{proposition}\label{lem:StructureNashLinear}
For linear utilities, the set of NE consists of all $(\delta_i)_{i \in N}$ with $\supp(\delta_i) \subseteq \Amax_i$ for all $i \in N$.
\footnote{
If we interpret the overall distribution as a probability distribution, \Cref{lem:StructureNashLinear} implies that a distribution is in equilibrium if and only if it is the outcome of a random dictatorship, where a player is selected uniformly at random, and then selects one of his best outcomes. Random dictatorship is sometimes used as a benchmark of fairness; see, e.g., \citet{BMS05a}.
}
\end{proposition}

\begin{proof}
By \Cref{def:Nasheq}, $(\delta^*_i)_{i \in N}$ is an NE if and only if $u_i(\delta^*)=\max_{\delta_i \in S_i} u_i(\delta^*-\delta^*_i+\delta_i)$ for all $i \in N$. By linearity of $u_i$, the latter expression is equivalent to $((n-1)/n)\cdot u_i(\delta^*-\delta^*_i)+(1/n)\max_{\delta_i \in S_i}u_i(n \cdot \delta_i)$. The term $u_i(n \cdot \delta_i)$ is maximal if and only if $\supp(\delta_i) \subseteq \Amax_i$, which proves the statement.   
\end{proof}

It is straightforward to see that the set of NEs is convex. Its simple structure allows for the computation of NEs with additional properties, e.g., the ones that maximize utilitarian welfare among equilibrium distributions.\footnote{Such equilibrium distributions might still be Pareto dominated by other (non-equilibrium) distributions \citep[see, e.g.,][]{BBP+19a}.} 

Furthermore, given a distribution $\delta$, it is easy to check whether it can be decomposed into an NE by solving the following linear feasibility program with variables $d_{i,x}$ (for $i \in N$ and $x \in A$).

\begin{align*}
\delta_x &= \sum_{i \in N}d_{i,x}              && \text{for all~} x \in A;
\\
1/n &=  \sum_{x \in \Amax_i} d_{i,x}             && \text{for all~} i \in N;
\\
d_{i,x} &\ge 0 && \text{for all~} i \in N,x \in A.
\end{align*}

\section{Leontief utilities}\label{sec:Leontief}

\citet{BGSS23b,BGSS23a} introduced \emph{Leontief preferences} for contexts where alternatives share characteristics of complements --- meaning that agents only benefit when they receive all of them (in a specific ratio), rather than being able to substitute one for another.
For each agent $i$, let $\Apos_i:= \{x\in A\colon v_{i,x}>0\}$ denote the set of alternatives the agent values positively.  
Leontief preferences are represented by
\begin{align*}
u_i(\delta)=\min_{x \in \Apos_i} \delta_x/v_{i,x}.
\end{align*}

These authors characterize the set of all NEs with the help of the concept of \emph{critical alternatives}.

\begin{definition}[\citet{BGSS23a}]
Given a distribution $\delta \in \Delta(A)$, agent $i$'s set of critical alternatives is defined as
\begin{align*}
T_{\delta,i} \coloneqq 
\arg \min_{x\in \Apos_i} \frac{\delta_x}{v_{i,x}}\text{.}
\end{align*}
\end{definition}
An alternative is critical in the sense that decreasing its share also decreases the utility of agent $i$. 

\begin{lemma}[\citet{BGSS23b,BGSS23a}, Lemma 1]
\label[lemma]{lem:eq-iff-critical}
In a budget-aggregation game where all agents have Leontief utilities, 
a profile $(\delta_i)_{i\in N}$ is an NE if and only if $\supp(\delta_i) \subseteq T_{\delta,i}$ for all $i \in N$
(where $\delta := \sum_{i\in N}\delta_i$).
\end{lemma}

\citeauthor{BGSS23b} further showed that all NEs result in the same overall distribution (called ``equilibrium distribution''), which can be computed by solving the convex program obtained from maximizing the product of agents' utilities. 
They leave the question open as to whether the equilibrium distribution can be computed exactly in polynomial time.
The remainder of this section is dedicated to answering that question in the affirmative.

\begin{theorem}
\label{thm:Leontiefeqcomp}
For Leontief utilities, a Nash equilibrium (and the corresponding unique equilibrium distribution) can be computed in time polynomial in the binary encoding length of the input.
\end{theorem}

Note that the equilibrium distribution does not change when individual valuations are rescaled. 
Similarly, rescaling contributions preserves the share each alternative receives. Thus, for the sake of simplicity, we assume throughout this section that all valuations are natural numbers.
We prove that the equilibrium distribution can be computed in time poly($n, m, \log_2(v_{\max})$) where $v_{\max} \coloneqq \max_{i\in N, x\in A}v_{i,x}$, i.e., the run-time is polynomial in the binary encoding length of the input.

As a first step, we prove an upper bound on the binary encoding length of the equilibrium distribution $\delta^*$ and its corresponding utility profile (which we denote by $u^*$).
\begin{lemma}
\label[lemma]{lem:edr-rational}
If the agents' valuations $v_{i,x}$ are natural numbers, then the equilibrium distribution $\delta^*$ and its utility profile $u^*$ are rational-valued.
~
Moreover, their binary encoding length is bounded by a polynomial function of the binary encoding length of $v_{i,x}$ and $n$.
\end{lemma}
\begin{proof}
For each $i\in N$, let $T_i$ be a non-empty set of alternatives. Consider the following linear program (LP), with variables $u_i$ (for $i \in N$), $d_x$ (for $x \in A$), and $d_{i,x}$ (for $i \in N$ and $x \in T_i$):
\begin{align*}
d_x &= u_i \cdot v_{i,x}              && \text{for all~} i \in N, x \in T_i;
\\
d_x &\geq  u_i \cdot v_{i,x}              && \text{for all~} i \in N, x \in A\setminus T_i;
\\
\sum_{x \in T_i} d_{i,x} &= 1/n       &&  \text{for all~}  i \in N;
\\
\sum_{i \in N} d_{i,x} &= d_x         &&   \text{for all~} x \in A;
\\
d_{i,x}&\geq 0         &&   \text{for all~} i \in N, x \in T_i\text.
\\
\end{align*}

Every solution of this LP (if any) represents a distribution $\delta_x=d_x$ for all $x \in A$, with a decomposition $\delta_{i,x}=d_{i,x}$ for all $i,x$, such that each agent $i$ contributes only to alternatives in $T_i$, and all alternatives in $T_i$ are critical for $i$. Variables $u_i$ take the values of the corresponding utilities. By \Cref{lem:eq-iff-critical}, such a strategy profile has to coincide with an NE.

An NE $(\delta_i^*)_{i \in N}$ is a solution to the above LP whenever $T_i = T_{\delta^*,i}$ for all $i\in N$. 
By assumption, the coefficients of this LP are all rational. Therefore, by well-known properties of linear programming, the LP has a rational solution,
with binary encoding length bounded by a polynomial function of the representation length of its coefficients.

We can give an explicit bound on the representation length of the equilibrium distribution $\delta^*$.
For that, we construct an undirected graph where vertices correspond to alternatives, and there is an edge between $x,y \in A$ if and only if there exists an agent $i \in N$ with $x,y \in T_{\delta^*,i}$. 
Each component of that graph can be considered separately, as the sum of contributions to that component equals the fraction of agents for whom parts of the component are critical alternatives. 

Thus, let $A' \subseteq A$ be a subset of $m'$ alternatives forming a component and $N' \subseteq N$ the subset of agents contributing to alternatives in $A'$. Choose  some $x_1 \in A'$.
If $m'=1$ (so $x_1$ is the unique alternative in $A'$,) then $\delta^*_{x_1}=N'/N$ as $x_1$ is the unique critical alternative for all agents in $N'$.
Otherwise, 
 there needs to be at least one other alternative $x_2$ that can be reached in one step from $x_1$, i.e., there exists an agent $i_1 \in N'$ such that $\delta^*_{x_1}/v_{i_1,x_1}=\delta^*_{x_2}/v_{i_1,x_2}$. 
Hence, $\delta^*_{x_2}=(v_{i_1,x_2}/v_{i_1,x_1}) \cdot \delta^*_{x_1}$. 
Similarly, if $m'\geq 3$ then there exists another alternative $x_3 \in A'$ that can be reached in one step from either $x_1$ or $x_2$. In general, given $k<m'$ connected alternatives from $A'$, there needs to be another one that can be reached in one step from one of the $k$ alternatives. 

This gives a system of linear equations where each $\delta^*_{x_j}$ can be written in terms of $\delta^*_{x_1}$.
Specifically, $\delta^*_{x_2}$ only requires the two valuations $v_{i_1,x_1}$ and $v_{i_1,x_2}$, $\delta^*_{x_3}$ requires at most four valuations
($v_{i_1,x_1},v_{i_1,x_2}$ 
and
$v_{i_2,x_2},v_{i_2,x_3}$ 
)
, and so on. Considering this ``worst case'' in terms of representation length together with $\sum_{x \in A'}\delta^*_{x}=\lvert N'\rvert/n$, 
$\delta^*_{x_1}$ can be written as a fraction with numerator $\lvert N'\rvert/n$ and denominator $1+v_{i_1,x_2}/v_{i_1,x_1}+(v_{i_1,x_2}/v_{i_1,x_1}) \cdot (v_{i_2,x_3}/v_{i_2,x_2})+ \dots$. This results in a binary encoding length of $\log_2((\lvert N'\rvert/n)\cdot{v_{\max}}^{m-1})$ for the numerator and $\log_2(m \cdot {v_{\max}}^{m-1})$ for the denominator, where we used $m'\leq m$, and $v_{i,x}\leq v_{\max}$.

As $x_1$ was chosen arbitrarily, and Leontief utilities coincide with distributions to certain alternatives, $(\delta^*,u^*)$ can be represented by $n+m$ times the derived length for $\delta^*_{x_1}$.
\end{proof}

\Cref{lem:edr-rational} cannot be used directly for computing the equilibrium distribution in polynomial time, since the proof requires us to know $T_{\delta^*,i}$; iterating over all possible $T_{\delta^*,i}$ would require exponential time.
To prove polynomial-time computability, we leverage the following lemma, which is Theorem 13 of \citet{Jain07a}:

\begin{lemma}[\citet{Jain07a}]
\label[lemma]{lem:jain}
Let $S$ be a convex set given by a strong separation oracle, and $\phi>0$ an integer.

There is an oracle-polynomial time and $\phi$-linear time algorithm which does one of the following:

\begin{enumerate}[noitemsep,topsep=0pt]
\item Concludes that there is no point in $S$ with binary encoding length at most $\phi$, or 
\item produces a point in $S$ with binary encoding length at most $P(n) \cdot \phi$,
where $P(n)$ is a polynomial.
\end{enumerate}
\end{lemma}

We apply \Cref{lem:jain} as follows.
For every positive rational number $z_0$, we define a convex set $S(z_0) \subseteq \mathbb{R}^{n+m}$, where the variables are $u_i$ for $i\in N$ and $d_x$ for $x\in A$:
\begin{align*}
\prod_{i=1}^n {u_i}^{1/n} &\geq z_0;
\\
d_x &\geq  u_i \cdot v_{i,x}  && \text{~for all~} i \in N, x \in A;
\\
\sum_{x \in A} d_{x} &= 1;
\\
u_i &\geq 0      &&   \text{~for all~} i \in N;
\\
d_x &\geq 0      &&   \text{~for all~} x \in A;
\end{align*}

Intuitively, the set $S(z_0)$ represents all pairs $(\delta,u)$ such that $\delta$ is a feasible distribution, 
$u$ is a lower bound on the corresponding utility profile, and the geometric mean of all utilities is at least $z_0$; see the proof of \Cref{thm:Leontiefeqcomp} below for a formal statement.
A \emph{strong separation oracle for $S(z_0)$} is a function that accepts as input a rational vector $\mathbf{y}' = (u'_1,\ldots,u'_n, d'_1,\ldots,d'_m)$. 
It should return one of two outcomes: either an assertion that $\mathbf{y}'\in S(z_0)$, or a hyperplane that separates $\mathbf{y}'$ from $S(z_0)$
(that is, a rational vector $\mathbf{c}$ such that $\mathbf{c}^\intercal \mathbf{y}' < \mathbf{c}^\intercal \mathbf{y}$ for all 
$\mathbf{y}\in S(z_0)$).

\begin{lemma}
\label[lemma]{lem:SzNull}
    For every rational $z_0>0$, there is a polynomial-time strong separation oracle for the convex set $S(z_0)$.
\end{lemma}
\begin{proof}
Given a rational vector $\mathbf{y}' = (u'_1,\ldots,u'_n, d'_1,\ldots,d'_m)$, we first check whether the point satisfies the linear constraints $d'_x \geq  u'_i \cdot v_{i,x}$,  $\sum_{x \in A} d'_{x} = 1$, $u'_i \geq 0$ and $d'_x \geq 0$. If one of these constraints is violated, the constraint itself yields a separating hyperplane. 
As the number of these constraints is polynomial in $n,m$, all of them can be checked in polynomial time.

It remains to handle the case that all linear constraints are satisfied, whereas the nonlinear constraint is violated, i.e., 
\begin{align}
\label{eq:violated-constraint}
\prod_{i=1}^n u'_i < {z_0}^n.
\end{align}
As $n \in \mathbb{N}$, the above condition can be checked exactly using arithmetic operations on rational numbers.
The binary encoding length of the product is polynomial in the binary encoding length of $u_i'$, meaning that the representation stays polynomial in the input size.

To construct a separating hyperplane for this case, we use an idea similar to the one in \cite{Jain07a}.

First, if $u'_i=0$ for some $i \in N$, choose $\mathbf{c}$ as the vector with value $1$ for each variable with $u'_i=0$ and value $0$ otherwise. By construction, $\mathbf{c}^\intercal \mathbf{y}' = 0$. For every $\mathbf{y} = (u_1,\ldots,u_n,d_1,\ldots,d_m) \in S(z_0)$, $\mathbf{c}^\intercal \mathbf{y}>0$ as $\prod_{i=1}^n (u_i) \geq {z_0}^n>0$.

Otherwise, we can define the vector $\mathbf{c}$ to have the coefficient 
$(1/n)\cdot (1/u_i')$ for each variable $u_i$, and the coefficient $0$ for each variable $d_x$.
Note that the encoding length of $\mathbf{c}$ is polynomial in the input size.
For every vector $\mathbf{y} = (u_1,\ldots,u_n,d_1,\ldots,d_m)$,
\begin{align*}
\mathbf{c}^\intercal \mathbf{y} = \frac{1}{n}\cdot \sum_{i =1}^n \frac{u_i}{u'_i}.
\end{align*}

Substituting $\mathbf{y} \coloneqq \mathbf{y}'$ gives 
$\mathbf{c}^\intercal \mathbf{y}' = 1$.
We now prove that 
$\mathbf{c}^\intercal \mathbf{y} > 1$
for every $\mathbf{y}\in S(z_0)$.

Indeed, $\mathbf{c}^\intercal \mathbf{y}$ is the arithmetic mean of the $n$ positive numbers $\frac{u_i}{u_i'}$.
By the AM-GM inequality, this sum is at least as large as their geometric mean, that is 
\begin{align*}
\frac{1}{n}\cdot \sum_{i =1}^n \frac{u_i}{u'_i} 
\;\geq\; 
\prod_{i=1}^n \left(\frac{u_i}{u_i'}\right)^{1/n}
\;=\;
\frac{(\prod_{i=1}^n (u_i)^{1/n})}
{(\prod_{i=1}^n (u'_i)^{1/n})}.
\end{align*}
Since $\mathbf{y} = (u_1,\ldots,u_n,d_1,\ldots,d_m)$ is in $S(z_0)$, it satisfies the inequality 
$\prod_{i=1}^n u_i \geq {z_0}^n$.
Substituting in the right-hand side above gives:
\begin{align*}
\frac{1}{n}\cdot \sum_{i =1}^n \frac{u_i}{u'_i} 
\;\geq\; 
\frac{{z_0}}
{\prod_{i=1}^n (u'_i)^{1/n}},
\end{align*}
which is larger than $1$ by inequality \eqref{eq:violated-constraint}.
Hence, $\mathbf{c}^\intercal \mathbf{y}>1$, so $\mathbf{c}$ indeed defines a separating hyperplane.
\end{proof}

\begin{proof}[Proof of \Cref{thm:Leontiefeqcomp}]

\Cref{lem:SzNull} 
allows us to apply \Cref{lem:jain} to $S(z_0)$.
We can now compute the equilibrium distribution by applying binary search to $z_0$, as follows.

\begin{enumerate}[noitemsep,topsep=0pt]
\item Initialize $L \coloneqq $ the product of agents' utilities for some arbitrary distribution (e.g., the uniform distribution).
\item Initialize $H \coloneqq $ some upper bound on the product of agents' utilities, 
e.g., the one resulting from the (unrealistic) distribution in which each agent $i$ divides the entire budget of $1$  optimally for him (in proportion to $v_i$). 
\item Let $\phi \coloneqq $ an upper bound on the binary encoding length of $(\delta^*,u^*)$, derived in  \Cref{lem:edr-rational}.
\item 
\label{step:z0}
Let $z_0 \coloneqq (L+H)/2$.
Note that both $L$ and $H$ can be encoded in length polynomial in the input size, so the same applies to $z_0$.

\item Apply \Cref{lem:jain} to $S(z_0)$, using \Cref{lem:SzNull} for the strong separation oracle. Consider the two possible outcomes.

\paragraph{Outcome 1} (``no point in $S(z_0)$ with binary encoding length at most $\phi$''): then the vector $(\delta^*,u^*)$ is not in $S(z_0)$. 
However, the vector $(\delta^*,u^*)$ satisfies all linear constraints in the definition of $S(z_0)$. Hence, it necessarily fails the non-linear  constraint. This means that the geometric mean of agents' utilities in the equilibrium distribution is below $z_0$.
Set $H \coloneqq z_0$ and return to Step \ref{step:z0}.

\paragraph{Outcome 2} (``a point in $S(z_0)$ with binary encoding length at most $P(n)\cdot \phi$''): then the $d_x$ variables in this point define a distribution $\delta$ by which the utility of each agent $i$ is at least $u_i$, and the geometric mean of all $u_i$ is at least $z_0$; hence, the geometric mean of the agents' utilities is at least $z_0$ too.

Check whether $\delta$ is the equilibrium distribution (this can be done in polynomial time using \Cref{lem:eq-iff-critical}). If it is, return $\delta$ and finish. Otherwise, set $L \coloneqq z_0$ and return to Step \ref{step:z0}.

\end{enumerate}
As the binary encoding length of $(\delta^*,u^*)$ is at most $\phi$, the binary encoding length of the maximum product of agents' utilities is at most $\sum_{i} (1/n) \log_2(u^*_i) \leq \phi$.
Therefore, after at most $\phi$ steps, the 
binary search is guaranteed to terminate with the equilibrium distribution.

Given the equilibrium distribution, one can easily compute an NE by turning the LP in the proof of \Cref{lem:edr-rational} into a linear feasibility problem where the variables $u_i$ and $d_x$ are replaced by their respective values in the equilibrium distribution. 
\end{proof}

Together with uniqueness of the overall distribution, this implies that the set of NEs is convex. 

Moreover, \Cref{thm:Leontiefeqcomp} is applicable to other utility models where NEs coincide with the ones for Leontief utilities, e.g., Cobb-Douglas utility functions \citep{balancedArXiv}.

\section{Binary symmetric separable utilities}
\label{sec:symsep}

In this section, we investigate NEs for very general utility models when agents have separable preferences. Separability requires that agents receive utility independently from each alternative. More formally, fixing the outcome on a subset of all alternatives has no influence on an agent's preferences over how to allocate the remaining budget among the other alternatives. The interested reader is advised to consult \citet{Debr59b,BPR78a} for an axiomatic characterization of separable preferences. Here, we say that preferences are separable if they can be represented by a utility function of the form
\begin{align*}
u_i(\delta)=\sum_{x \in A}g_{i,x}(\delta_x)
\end{align*}
where $g_{i,x} \colon \mathbb{R}_{\geq 0} \to \mathbb{R}$ are continuous and weakly increasing.

Linear utilities represent the most common class of separable preferences. Contrarily, Leontief preferences are not separable as the minimum is taken over all alternatives with positive valuations.

We consider the class of \emph{binary symmetric} separable preferences where each agent $i$ has a non-empty set of approved alternatives $A_i \subseteq A$ but does not further distinguish between the alternatives.\footnote{This is what we mean by the term ``symmetric'' in our context.} Furthermore, her utility is assumed to \textit{strictly} increase when the budget on some $x \in A_i$ increases.

In terms of utility functions, such preferences are represented by
\begin{align*}
u_i(\delta) = \sum_{x \in A_i}g_i(\delta_x) 
\end{align*}
where $g_i \colon [0,1] \to \mathbb{R}$ are continuous and strictly increasing functions.

For example, setting $g_i(\delta_x)=\delta_x$ for all $i \in N$ corresponds to linear preferences with binary valuations $v_{i,x} \in \{0,1\}$ \citep[see, e.g.,][]{BMS05a}. Choosing $g_i(\delta_x)=\ln(\delta_x)$ leads to Cobb-Douglas utilities with binary weights. However, note that our model allows the $g_i$'s to differ among agents.
~
The following example shows that for such general preferences, NEs do not always exist.

\begin{example}
Consider an instance with $n=2$ and $A=\{a,b\}$, with $g_1(\cdot) = \sqrt{\cdot}$, i.e., $u_1(\delta)=\sqrt{\delta_a}+\sqrt{\delta_b}$ and $g_2(\cdot) = \cdot^2$, i.e, $u_2(\delta)=\delta_a^2+\delta_b^2$. 

If $\delta_a>\delta_b$, then $\delta_a>1/2$, hence both agents contribute to $a$.  
By strict concavity of $g_1$, Agent $1$ has an incentive to move some budget from $a$ to $b$. 
The case $\delta_a<\delta_b$ can be treated analogously.    

If $\delta_a=\delta_b$, then by strict convexity of $g_2$, Agent $2$ has an incentive to move her entire budget to an alternative $x \in \{a,b\}$ with $\delta_{1,x}\ge 1/4$, resulting in $\delta_a \neq \delta_b$. Thus, there cannot exist an NE. 

\end{example}

Intuitively, Agent $1$ wants the entire budget to be distributed as equally as possible, whereas Agent $2$ prefers to focus on one alternative.
~
Therefore, we treat separately the two cases where all $g_i$'s are strictly concave (\Cref{sec:strictlyconcave}) and strictly convex (\Cref{sec:strictlyconvex}).

\subsection{Strictly concave utilities} \label{sec:strictlyconcave}
In this section, all $g_i$'s are assumed to be strictly concave, i.e., $g_i(\lambda\delta+(1-\lambda)\delta')>\lambda g_i(\delta)+(1-\lambda)g_i(\delta')$ for all different $\delta,\delta' \in \Delta(A)$ and $\lambda \in (0,1)$.

\begin{theorem}[\citet{BGSS23a}, Proposition 9] \label{thm:Leontiefstrictly concave}
For binary symmetric separable and strictly concave utilities, the set of NEs coincides with that for Leontief utilities with the same valuations ($v_{i,x} \in \{0,1\}$).
\end{theorem}

For binary valuations, \citet{BGSS23a} stated two different possibilities to compute an overall equilibrium distribution via linear programming. This distribution can be decomposed into an NE using the same linear program described at the end of \Cref{thm:Leontiefeqcomp}.

\begin{corollary}
For binary symmetric separable and strictly concave utilities, an NE can be computed via linear programming.
\end{corollary}

Furthermore, the concept of critical alternatives can readily be transferred in order to check whether a given strategy profile constitutes an NE.

\subsection{Strictly convex utilities} \label{sec:strictlyconvex}

In this section, all $g_i$'s are assumed to be strictly convex, i.e., $\lambda g_i(\delta)+(1-\lambda)g_i(\delta')>g_i(\lambda\delta+(1-\lambda)\delta')$ for all different $\delta,\delta' \in \Delta(A)$ and $\lambda \in (0,1)$.

Note that in this case, each utility function $u_i(\delta)$, seen as an $ m$-dimensional real-valued function, is convex, as it is a sum of strictly convex functions. Furthermore, it is strictly convex if the domain is restricted to the $\lvert A_i \rvert$-dimensional subspace spanned by agent $i$'s approved alternatives. 

Let $M_{\delta,i} \coloneqq 
\arg \max_{x\in A_i}\delta_x$ be the set of alternatives with maximum share under the distribution $\delta$ that are approved by agent $i$.

\begin{lemma}
\label{lem:convexeq}
For binary symmetric separable and strictly convex utilities, the set of Nash equilibria consists of all $(\delta^*_i)_{i\in N}$ with $\supp(\delta^*_i) = M_{\delta^*,i}$ and $\lvert\supp(\delta^*_i)\rvert =\lvert M_{\delta^*,i}\rvert =1$ for all $i\in N$.     
\end{lemma}
\begin{proof}
The distribution $(\delta^*_i)_{i \in N}$ is an NE if and only if every agent $i$'s distribution $\delta^*_i$ is a best response and therefore a solution of the optimization problem 
\begin{align*}
\max_{\delta_i \in S_i} u_i(\delta^*-\delta^*_i+\delta_i)\text{.}
\end{align*}
The objective function is strictly convex on this domain, and the feasible region is the ($\lvert A_i\rvert -1$)-dimensional simplex spanned by those strategies where agent $i$ places her entire budget on a single approved alternative. So every maximizer is a vertex of this simplex, which proves that $\lvert\supp(\delta_i)\rvert=1$. 
By adding $1/n$ to an alternative $x\in A_i$ in $\delta^*-\delta^*_i$, agent $i$'s utility increases by $g_i(\delta^*_x-\delta^*_{i,x}+1/n)-g_i(\delta^*_x-\delta^*_{i,x})$. This increase is maximal if and only if $x$ is an alternative with maximal share under the distribution $\delta^*-\delta^*_i$, i.e., $x\in M_{\delta^*-\delta^*_i,i}$ since $g_i$ is strictly convex.
This proves that $(\delta^*_i)_{i \in N}$ is an NE if and only if every agent plays a pure strategy, i.e., $\lvert\supp(\delta^*_i)\rvert =1$, and $\supp(\delta^*_i) \subseteq M_{\delta^*-\delta^*_i,i}$ for all $i\in N$. Thus, $\lvert M_{\delta^*,i}\rvert=1$ and $\supp(\delta^*_i) = M_{\delta^*,i}$ for all $i\in N$. 

For the other direction, it remains to show that if $\supp(\delta_i) = M_{\delta,i}$ and $\lvert M_{\delta,i}\rvert =1$ for all $i\in N$, it holds that $\supp(\delta_i) \subseteq M_{\delta-\delta_i,i}$ for all $i\in N$. Therefore, assume $x\in \supp(\delta_i)$ and $y\in A \backslash \lbrace x\rbrace$ for agent $i$. Then, since all agents play pure strategies, all $\delta_x$ are some multiple of $1/n$ and $\delta_x\geq\delta_y+1/n$. Thus,  
\begin{align*}
\delta_y-\delta_{i,y}=\delta_y\leq\delta_x-1/n=\delta_x-\delta_{i,x}
\end{align*}
and $\supp(\delta_i) \subseteq M_{\delta-\delta_i,i}$.
\end{proof}

\begin{corollary}\label{cor:equalbudget}
For a Nash equilibrium $(\delta^*_i)_{i \in N}$ and two alternatives $x,y \in A$, $\delta^*_x=\delta^*_y$ implies $\lvert A_i \cap \{x,y\}\rvert=1$ for all $i \in N$ with $\supp(\delta^*_i) \in \{x,y\}$.
\end{corollary}

\Cref{lem:convexeq} implies that checking whether a given distribution $\delta$ admits an NE $(\delta_i)_{i\in N}$ can be done in polynomial time. In a first step, verify that $\lvert M_{\delta,i}\rvert=1$ for all $i\in N$. Second, check that the decomposition $(\delta_i)_{i\in N}$ with $\delta_{i,x}=1/n$ for $x\in M_{\delta,i}$ (and $\delta_{i,x}=0$ otherwise) is a valid decomposition of $\delta$.

\begin{theorem}
    For binary symmetric separable and strictly convex utilities, a Nash equilibrium always exists and can be found in polynomial time.
\end{theorem}

\begin{proof}
We use the following greedy approach to construct an NE $(\delta_i^*)_{i \in N}$. 
\footnote{
This algorithm is equivalent to an algorithm called \emph{``maximum payment rule''} \citep{ALL+25b}, which was recently presented as a rule for approximately-fair budget distribution for agents with binary linear utilities.
}
Let $N'\subseteq N$ denote the set of agents who have not yet allocated their budgets.
In each step, select an alternative $x^* \in \arg \max_{x \in A} \sum_{i \in N': x \in A_i} 1/n$, i.e., $x^*$ maximizes the aggregated share of the remaining agents who approve it. Each of these agents then allocates her entire budget to $x^*$, and is excluded from $N'$. The procedure is repeated until no agents remain.

This process runs in polynomial time in $n$ and $m$ as it terminates after at most $n$ steps with a complete strategy profile $(\delta^*_i)$ where in each step, a suitable alternative $x^*$ can be found by counting approvals of the remaining agents in $N'$. By construction, $\lvert\supp{\delta^*_i}\rvert=1$ for all $i \in N$. Furthermore, $\supp(\delta^*_i) \subseteq M_{\delta^*,i}$. Otherwise, some agent $i$ approves an alternative $x$ with $\delta^*_x>\delta^*_{\supp(\delta^*_i)}$ and would have contributed to $x$ instead under the described greedy approach.

Finally, if $\lvert M_{\delta^*,i}\rvert>1$ for some agent $i$, then there exist alternatives $x,y \in A_i$ with $\delta^*_x = \delta^*_y$ and w.l.o.g. $\supp(\delta^*_i)=x$, contradicting \Cref{cor:equalbudget}. 

Thus, $(\delta^*_i)_{i \in N}$ constitutes an NE by \Cref{lem:convexeq}.    
\end{proof}

The following example shows that the set of NEs is not convex in general.

\begin{example}
Consider a profile with a single agent $i$ and two alternatives $x$ and $y$ that are both approved by the agent, i.e., $A_i=A$. Using \Cref{lem:convexeq}, it is easy to see that there exist exactly two equilibria, which are given by the two distributions that allocate the entire budget to one of the alternatives.  
\end{example}

In fact, the number of NE is finite and can be bounded explicitly.

\begin{lemma}
The size of the set of Nash equilibria is bounded by $\min\{m!,\prod_{i\in N}\lvert A_i\rvert\}$     
\end{lemma}

\begin{proof}
First, there are $m!$ possibilities to strictly order the set of alternatives by the amount of budget they receive. For each ordering, there can exist at most one NE for which the total distribution follows this ordering from highest to lowest share placed on the alternatives. Ties can be broken arbitrarily due to \Cref{cor:equalbudget}. By \Cref{lem:convexeq}, the ordering uniquely determines the sets $M_{\delta,i}$ and thus, the possible corresponding NE.

Second, \Cref{lem:convexeq} shows that in any NE, every agent always plays a pure strategy. This implies that each agent $i$ has at most $\lvert A_i\rvert$ different best responses. Therefore, there are $\prod_{i\in N}\lvert A_i\rvert$ possible combinations of best responses, i.e., potential NEs.
\end{proof}

Both of these upper bounds are sharp for certain profiles.

\begin{example}
First, consider a profile with $m$ alternatives and $n:=2^m-1$ agents, one agent for every possible set of approved alternatives. 
In particular, each alternative is approved by half of the agents. Then, for every possible strict ordering of alternatives $x_1 \succ \dots \succ x_n$, there exists an equilibrium with shares following this ordering: alternative $x_\ell$ receives the budgets of all agents $i$ with $x_\ell \in A_i$ and $x_k \not\in A_i$ for all $k <\ell$, i.e., $2^{m-1}/2^{\ell-1}=2^{m-\ell}$ agents.  
Thus, the allocated shares in decreasing order are given by $2^{m-1}/n,2^{m-2}/n,...,1/n$.

Second, for every set of positive integers $a_1,\dots,a_n$, 
consider a profile with $n$ agents 
and $m:=a_1+\cdots+a_n$ alternatives, 
such that each agent $i$ approves a set $A_i$ of $a_i$ alternatives, and all $A_i$'s are pairwise disjoint.
Then, every possible combination of strategies $\delta_i$ with $\lvert \supp(\delta_i)\rvert=1$ and $\supp(\delta_i) \in A_i$ results in an equilibrium, giving $\prod_{i=1}^n \lvert A_i \rvert$ equilibria in total.
\end{example}

\section{$\ell_1$ disutilities}
\label{sec:l1_preferences}

A widely studied utility model in the context of budget-aggregation is based on $\ell_1$ preferences, also referred to as $\ell_1$ disutilities \citep{Lind11a, GKSA19a, FPP+19a}.
The existing literature mostly focuses on the axiomatic investigation of mechanisms \citep{BGSS24a,EGL+26a}, in particular, moving phantom mechanisms \citep{FPPV21a,BFS+24a}. 

For $\ell_1$ preferences, each agent is assumed to have a peak allocation — a preferred distribution of the budget — and her utility is determined by the negative $\ell_1$ distance between any given allocation and her peak.
Formally, let $p_i=(p_{i,x})_{x \in A} \in \Delta(A)$ denote the peak of agent $i$, then her utility for a given allocation $\delta \in \Delta(A)$ is defined as  
\begin{equation*}
u_i(\delta) = - \lVert \delta - p_i \rVert_1.
\end{equation*}

Alternatively, the utility can be interpreted in terms of overlap utilities, as shown by \citet{GKSA19a}. 
In the overlap utility model, the utility of agent $i$ is defined as the sum of the overlaps between the allocation 
$\delta$ and the agent's peak $p_i$:
\begin{equation*}
u_i(\delta) = \sum_{x \in A} \min(\delta_x , p_{i,x}).
\end{equation*}
\citet{GKSA19a} demonstrate that $\ell_1$ disutilities are equivalent (up to an affine transformation) to overlap utilities, making both models interchangeable in our budget-aggregation context. Note that these preferences are separable, but not binary symmetric separable.

These utilities capture the behavior of agents who derive value from receiving resources up to a certain satiation point, beyond which additional resources provide no further benefit. 
Specifically, the utility is linear in the allocation up to the agent’s peak, and constant thereafter — reflecting satiation once the desired threshold is met.

We say an alternative $x$ is \textit{oversupplied} for agent $i$ if $\delta_x > p_{i,x}$.
Additionally, let $M_i^{>}(\delta)$ be the set of alternatives that are oversupplied for agent $i$ given an allocation $\delta$.
Similarly, an alternative $x$ is \textit{undersupplied} for agent $i$ if $\delta_x < p_{i,x}$ and $M_i^{<}(\delta)$ is the set of undersupplied alternatives. Define $M_i^{\leq}(\delta^*)$ accordingly.
We show that a strategy profile constitutes an NE when no agent contributes to one of her oversupplied alternatives. 

\begin{lemma} \label{lem:l1_nash_condition}
For $\ell_1$ preferences, a strategy profile $(\delta^*_i)_{i \in N}$ is a Nash equilibrium if and only if  $\supp(\delta^*_i) \subseteq M_i^{\leq}(\delta^*)$ for all $i \in N$.
\end{lemma}

\begin{proof}
$\Rightarrow$: Assume that there is some agent $i$ with $\delta^*_{i, x} > 0$ for some alternative $x$ which is oversupplied, i.e. $\delta^*_x > p_{i, x}$.
Redistributing $\delta^*_{i, x}$ to some undersupplied alternative, say $x'$, strictly increases the utility of agent $i$, showing that the strategy profile $(\delta^*_i)_{i \in N}$ is not an NE. 

$\Leftarrow$: Suppose $(\delta^*_i)_{i \in N}$ satisfies $\supp(\delta^*_i) \subseteq M_i^{\leq}(\delta^*)$ for all $i \in N$.
Since the marginal utility decrease in $x$ equals the marginal utility increase in $x'$ for all $x, x' \in M_i^{\leq}(\delta^*)$, agent $i$ cannot strictly increase her utility by redistributing her budget $\delta_i$.
Therefore, $u_i(\delta^*)=\max_{\delta_i' \in S_i} u_i(\delta^*-\delta^*_i+\delta_i')$, establishing that $(\delta^*_i)_{i \in N}$ is an NE.
\end{proof}

The existence of an NE for $\ell_1$ disutilities follows from the more general existence result by \citet{Debr52a} for quasi-concave utilities. 
Given our above characterization, it is not surprising that NEs do not yield a unique overall distribution in the case of $\ell_1$ preferences, since the characterization depends solely on the support of the agents' individual distributions.

\begin{example}\label{example:l1_nash_notUnique}
Let $n=3$ and consider the following profile 
\begin{equation*}
P =
\begin{pmatrix}
    3/10 & 3/10 & 0 & 4/10 \\
    2/10 & 2/10 & 6/10 & 0 \\
    2/10 & 2/10 & 6/10 & 0 \\
\end{pmatrix},
\end{equation*}
where each row denotes the peak of an agent.
Since this instance has three agents, everyone has a budget of $1/3$.
Both $\delta = (3/10, 2/10, 14/30, 1/30)$ and $\delta' = (2/10, 3/10, 14/30, 1/30)$ admit decompositions into strategies that satisfy \Cref{lem:l1_nash_condition}, i.e. no agent contributes to oversupplied alternatives:
\begin{align*}
\left(\delta_i \right)_{i \in N} = 
\begin{pmatrix}
    3/10 & 0 & 0 & 1/30 \\
    0 & 2/10 & 4/30 & 0 \\
    0 & 0 & 10/30 & 0 \\
\end{pmatrix}
\\
\left(\delta'_i \right)_{i \in N} = 
\begin{pmatrix}
    0 & 3/10 & 0 & 1/30 \\
    2/10 & 0 & 4/30 & 0 \\
    0 & 0 & 10/30 & 0 \\
\end{pmatrix}
\end{align*}
This verifies that both agent profiles $(\delta_i)_{i \in N}$ and $(\delta_i')_{i \in N}$ are NEs for the given set of peaks $P$.
\end{example}

With the above example, we can show the following.

\begin{proposition}
For $\ell_1$ preferences, the set of equilibrium distributions is not convex.
\end{proposition}

\begin{proof}
For the proof, consider again \Cref{example:l1_nash_notUnique} and a convex combination of the two overall distributions $\delta$ and $\delta'$.
More specifically, let $\delta'' := \frac{1}{2}\delta + \frac{1}{2}\delta'$.
We get $\delta'' = (1/4, 1/4, 14/30, 1/30)$.
Since Agents $2$ and $3$ only support the first and second alternative up to $2/10$, every strategy profile that is in equilibrium must have $\delta_{i,x} = 0$ for $i \in \{2, 3\}$ and $x \in \{1, 2\}$.
Since the total budget of the first agent is only $1/3$, she cannot cover both alternatives simultaneously.
Consequently, $\delta''$ cannot be decomposed into a strategy profile that is also an NE.
This establishes that the set of equilibrium distributions is not convex.
\end{proof}

Although the set of NEs may be non-convex --- demonstrating potential structural complexity --- there nonetheless exists a polynomial-time algorithm that computes an NE.

\begin{theorem}
For $\ell_1$ preferences, a Nash equilibrium can be computed in polynomial time.   
\end{theorem}

\begin{proof}
We present a polynomial-time algorithm.
Let each agent $i$ have a budget, denoted by $b_i$.
Initially, this budget is set to $b_i = 1/n$ for all $i \in N$.
Iterate over all alternatives $x \in A$ in an arbitrary order.
For a given alternative $x$, order all agents in a descending order of their peaks for this alternative, i.e., $p_{i_1, x} \geq p_{i_2, x} \geq \dots \geq p_{i_n,x}$.
When it is agent $i$'s turn, she allocates $\min(b_i, \max(0, p_{i,x}-\delta_x))$ on alternative $x$.
Afterwards, $b_i$ is decreased and $\delta_x$ is increased by this amount.
Once agent $i$ has spent her entire share ($b_i=0$), she is deleted from the set of considered agents.
The ordering of the agents ensures that $x$ does not become oversupplied to an agent who already contributed to $x$ later on, i.e., no agent contributes to an oversupplied alternative.
When no agent wants to contribute to alternative $x$, we go to the next alternative $x'$ and reorder the remaining agents according to their peaks ($p_{i_1, x'} \geq p_{i_2, x'} \dots \geq p_{i_n,x'}$).
Then, we let them again sequentially contribute to this alternative.

We prove that every agent spends her entire budget after iterating over all alternatives.
Assume for contradiction that there is some agent who has still some budget left after she had the possibility to contribute to every alternative.
This would mean that $p_{i, x} \leq \delta_x$ for all $x \in A$ in the final partial distribution $\delta$, which is impossible due to the fact that $p_i \in \Delta(A)$ but $\sum_{x \in A}\delta_x<1$.
All in all, the suggested construction returns a strategy profile that constitutes an NE. 

Iterating over all $m$ alternatives and, for every alternative, sorting at most $n$ agents according to their peaks leads to a run-time complexity of $O(m n \log n)$.
\end{proof}

This construction provides an efficient algorithm for obtaining an NE given a set of optimal distributions of the agents.
Additionally, we can also show whether a given overall distribution can be represented as individual distributions that constitute an NE.

\begin{theorem}
For $\ell_1$ preferences, one can check in polynomial time whether an overall distribution can be decomposed into a Nash equilibrium.
\end{theorem}

\begin{proof}
Let $\delta$ be the overall distribution. 
We now consider the following linear constraints to find individual distributions constituting an NE where $d_{i, x}$ are the variables for all $i \in N$ and $x \in A$:
\begin{align*}
\sum_{x \in A}d_{i, x} \leq \frac{1}{n} \text{ for all } i \in N, \\
\sum_{i \in N}d_{i, x} = \delta_x \text{ for all } x \in A, \\
d_{i, x} (p_{i, x} - \delta_x) \geq 0  \text{ for all } i \in N \text{ and } x \in A,\\
d_{i, x} \geq 0 \text{ for all } i \in N \text{ and } x \in A.
\end{align*}
Let $d_{i, x}$ be the individual distribution of agent $i$ for alternative $x$, i.e., $d_{i, x} = \delta_{i,  x}$ for all $i \in N$ and $x \in A$.
The third constraint requires that an agent never contributes to an oversupplied alternative, which is a necessary and sufficient condition for a set of strategy profiles to form an NE (see \Cref{lem:l1_nash_condition}).
Furthermore, the second constraint ensures that individual strategies $(\delta_{i})_{i \in N}$ form the overall distribution $\delta$. 
Therefore, if these constraints are satisfied, the individual strategies add up to the overall distribution, which is in turn an equilibrium distribution due to the third constraint.
\end{proof}

\section{Weighted agents}
\label{sec:weighted}
In practice, one might want to assign different shares to the agents, e.g., when they represent groups of citizens of different sizes, or citizens in different tax brackets.%

We define a 
\emph{weighted budget-aggregation game} as a tuple $(N,A,\allowbreak(u_i)_{i \in N},(w_i)_{i\in N})$,
where $N$, $A$ and $(u_i)_{i \in N}$ are as in \Cref{def:game}, and $w_i$ is the weight of agent $i$, with $\sum_{i\in N}w_i = 1$.
The strategy set of each agent $i$ is given by
\begin{align*}
S_i(w_i) := \{\delta_i \in [0,1]^m \text{ such that } \sum_{x \in A}\delta_{i,x}=w_i\}\text{.}
\end{align*}
The definition of an NE transfers directly from \Cref{def:Nasheq}, except that $S_i$ depends on $w_i$ as stated above.

For some classes of utility functions, the weighted setting can be reduced to the unweighted setting by agent cloning if all weights are rational numbers, resulting in a pseudo-polynomial time algorithm. We formalize the agent cloning process below.
\begin{definition}
Given a weighted budget-aggregation game  $(N,A,\allowbreak(u_i)_{i \in N},(w_i)_{i\in N})$ where all $w_i$ are rational numbers,
let $D$ be their least common denominator, such that $w_i = q_i/D$ for all $i\in N$, and all $q_i$ are integers.
The corresponding \emph{cloned game} is an (unweighted) budget-ag\-gre\-gation game in which the set of agents $N^c$ contains $q_i$ clones of each agent $i$ (with the same utility function $u_i$), so that $\displaystyle |N^c|=\sum_{i\in N}q_i = D$.
\end{definition}

\begin{lemma}\label[lemma]{lem:clones}
    Let $(N,A,(u_i)_{i \in N},(w_i)_{i\in N})$ be a weighted bud\-get-ag\-gre\-gation game in which all $w_i$ are rational numbers,
	and all $u_i$ are either linear, Leontief, binary symmetric separable and strictly concave/convex, or $\ell_1$.  
    Let $(\delta_{i^c})_{i^c\in N^c}$ be a Nash equilibrium of the corresponding cloned game.
    Define a strategy profile $(\delta_{i})_{i\in N}$ such that $\delta_i$ is the sum of $\delta_{i^c}$ for all $i^c$ which are clones of $i$.
    Then $(\delta_{i})_{i\in N}$ is a Nash equilibrium of the weighted game.
\end{lemma}

\begin{proof}
Suppose that $(\delta_{i})_{i\in N}$ is not an NE of the weighted game.
For linear, Leontief, and $l_1$ preferences, there is a structural lemma that implies that, under $\delta$, some agent $i$ contributes a positive amount to some project $y$ in a set of ``forbidden'' projects that depends only on $u_i$ and $\delta$. 
Hence, under $\delta$, some clone $i_c$ of $i$ must contribute a positive amount to the same project $y$.
As $i_c$ and $i$ have the same utility function, the set of ``forbidden'' projects under $\delta$ is the same for both of them; as $i_c$ contributes to a ``forbidden'' project,   $\delta$ does not correspond to an NE of the cloned game.

Specifically:
\begin{itemize}
	\item For linear utilities, 
	by \Cref{lem:StructureNashLinear},
	some agent $i$ contributes under $\delta$ to some non-highest-valued project $y \in A\setminus \Amax_i$;
	\item For Leontief utilities, by \Cref{lem:eq-iff-critical}, some agent $i$  contributes under $\delta$ to some non-critical project $y\in A\setminus  T_{\delta,i}$;
	\item For $\ell_1$ disutilities, by \Cref{lem:l1_nash_condition}, some agent $i$  contributes under $\delta$ to some over-supplied project $y\in M^{>}_i(\delta)$. 
\end{itemize}

It remains to prove the lemma for binary symmetric separable and strictly concave/convex preferences. 

\Cref{thm:Leontiefstrictly concave} remains true for the weighted game \citep{BGSS23a}. Thus, some agent in the cloned game has to contribute to a non-critical alternative under $(\delta_{i})_{i \in N}$ in the cloned game and has a beneficial deviation, showing that the statement also holds for binary symmetric separable and strictly concave preferences. 

Finally, if $(\delta_{i})_{i\in N}$ is no NE of the weighted game for binary symmetric separable and strictly convex preferences, this means that some agent $i$ contributes either to an alternative outside of $M_{\delta,i}$ or to two alternatives $x,y \in M_{\delta,i}$. For the first case, this implies that there exists a clone of agent $i^c$ that contributes to an alternative outside of $M_{\delta,i^c}$ in the cloned game. \Cref{cor:equalbudget} proves that also for the second case, $(\delta_i)_{i \in N}$ cannot refer to an NE of the cloned game.\qedhere
\end{proof}

\begin{theorem}
\label[theorem]{thm:weighted}
A Nash equilibrium can be computed in pseudo-polynomial time
for weighted budget-aggregation games in which all utilities are 
either linear, Leontief, binary symmetric separable and strictly concave/convex, or $\ell_1$. 
\end{theorem}
\begin{proof}
Given a weighted budget-aggregation game with weights $w_i = q_i / D$, construct the corresponding cloned game, which is an unweighted game. This game has $D$ agents; hence, an NE can be computed in time polynomial in $m$ and $D$. 
By \Cref{lem:clones}, this NE corresponds to an NE for the weighted game; the run-time is pseudo-polynomial in the parameters of that game.
\end{proof}

We believe \Cref{lem:clones} can be extended to other classes of utilities. It would be interesting to characterize the set of utility functions for which \Cref{lem:clones} remains true. 
Moreover, it is open whether an NE in a weighted budget-aggregation game can be computed in time polynomial in $n$ and $m$ (rather than pseudo-polynomial), e.g., for Leontief preferences.

\section{Conclusion}
In this paper, we introduced budget-aggregation games to formalize a game-theoretic distributed approach for dividing a budget among a set of public projects and investigated NEs of the resulting games. Although finding NEs is known to be hard for general classes of utility models, our work shows that in practice, they can be computed efficiently. In detail, we proved that for the most common preference models of linear, Leontief, and $\ell_1$ preferences as well as for classes of binary symmetric and separable preferences, NEs can be computed in polynomial time.

It would be interesting to see whether there exists a class of utility functions containing all of the previously considered preference models that still allow for an efficient computation of NEs. As a first step, one might want to extend binary symmetric separable utilities to allow for valuations other than $0$ or $1$, or consider preferences based on $\ell_p$ metrics other than $\ell_1$.

Other possible directions for future work include the selection of specific equilibria, e.g., by incorporating additional axioms like Pareto efficiency or investigating limits of equilibrium dynamics, and the investigation of the price of anarchy.

\section*{Acknowledgments}
We thank Yakov Babichenko, Felix Brandt, Warut Suksompong, and the anonymous reviewers of AAMAS 2026 for helpful comments. Patrick is supported by the Deutsche Forschungsgemeinschaft under grant BR 2312/14-1.
Erel is supported by the Israel Science Foundation grant no. 1092/24.

\bibliographystyle{abbrvnat}

\begin{thebibliography}{98}
\providecommand{\natexlab}[1]{#1}
\providecommand{\url}[1]{\texttt{#1}}
\expandafter\ifx\csname urlstyle\endcsname\relax
  \providecommand{\doi}[1]{doi: #1}\else
  \providecommand{\doi}{doi: \begingroup \urlstyle{rm}\Url}\fi

\bibitem[Ackermann et~al.(2008)Ackermann, R{\"o}glin, and V{\"o}cking]{ackermann2008impact}
H.~Ackermann, H.~R{\"o}glin, and B.~V{\"o}cking.
\newblock On the impact of combinatorial structure on congestion games.
\newblock \emph{Journal of the ACM (JACM)}, 55\penalty0 (6):\penalty0 1--22, 2008.

\bibitem[Adsul et~al.(2011)Adsul, Garg, Mehta, and Sohoni]{adsul2011rank}
B.~Adsul, J.~Garg, R.~Mehta, and M.~Sohoni.
\newblock Rank-1 bimatrix games: A homeomorphism and a polynomial time algorithm.
\newblock In \emph{Proceedings of the 43rd Annual ACM Symposium on Theory of Computing (STOC)}, pages 195--204, 2011.

\bibitem[Arrow and Debreu(1954)]{ArDe54a}
K.~J. Arrow and G.~Debreu.
\newblock Existence of an equilibrium for a competitive economy.
\newblock \emph{Econometrica}, 22\penalty0 (3):\penalty0 265--290, 1954.

\bibitem[Aziz and Shah(2021)]{AzSh20a}
H.~Aziz and N.~Shah.
\newblock Participatory budgeting: Models and approaches.
\newblock In T.~Rudas and P.~G\'{a}bor, editors, \emph{Pathways between Social Science and Computational Social Science: Theories, Methods and Interpretations}. Springer, 2021.

\bibitem[Aziz et~al.(2020)Aziz, Bogomolnaia, and Moulin]{ABM20a}
H.~Aziz, A.~Bogomolnaia, and H.~Moulin.
\newblock Fair mixing: the case of dichotomous preferences.
\newblock \emph{ACM Transactions on Economics and Computation}, 8\penalty0 (4):\penalty0 18:1--18:27, 2020.

\bibitem[Aziz et~al.(2025)Aziz, Lederer, Lu, Suzuki, and Vollen]{ALL+25b}
H.~Aziz, P.~Lederer, X.~Lu, M.~Suzuki, and J.~Vollen.
\newblock Approximately fair and population consistent budget division via simple payment schemes.
\newblock In \emph{Proceedings of the 26h ACM Conference on Economics and Computation (ACM-EC)}, page 349, 2025.

\bibitem[Babichenko(2018)]{babichenko2018fast}
Y.~Babichenko.
\newblock Fast convergence of best-reply dynamics in aggregative games.
\newblock \emph{Mathematics of Operations Research}, 43\penalty0 (1):\penalty0 333--346, 2018.

\bibitem[Bhaskar and Lolakapuri(2018)]{bhaskar2018equilibrium}
U.~Bhaskar and P.~R. Lolakapuri.
\newblock Equilibrium computation in atomic splittable routing games.
\newblock In \emph{26th Annual European Symposium on Algorithms (ESA 2018)}, pages 58--1. Schloss Dagstuhl--Leibniz-Zentrum f{\"u}r Informatik, 2018.

\bibitem[Bil{\`o} and Mavronicolas(2014)]{bilo2014complexity}
V.~Bil{\`o} and M.~Mavronicolas.
\newblock Complexity of rational and irrational {N}ash equilibria.
\newblock \emph{Theory of Computing Systems}, 54\penalty0 (3):\penalty0 491--527, 2014.

\bibitem[Blackorby et~al.(1978)Blackorby, Primont, and Russell]{BPR78a}
C.~Blackorby, D.~Primont, and R.~R. Russell.
\newblock \emph{Duality, Separability and Functional Structure: Theory and Economic Applications}.
\newblock New York: North-Holland, 1978.

\bibitem[Bogomolnaia et~al.(2005)Bogomolnaia, Moulin, and Stong]{BMS05a}
A.~Bogomolnaia, H.~Moulin, and R.~Stong.
\newblock Collective choice under dichotomous preferences.
\newblock \emph{Journal of Economic Theory}, 122\penalty0 (2):\penalty0 165--184, 2005.

\bibitem[Brandl et~al.(2022)Brandl, Brandt, Greger, Peters, Stricker, and Suksompong]{BBP+19a}
F.~Brandl, F.~Brandt, M.~Greger, D.~Peters, C.~Stricker, and W.~Suksompong.
\newblock Funding public projects: {A} case for the {N}ash product rule.
\newblock \emph{Journal of Mathematical Economics}, 99:\penalty0 102585, 2022.

\bibitem[Brandt et~al.(2016)Brandt, Conitzer, Endriss, Lang, and Procaccia]{BCE+15a}
F.~Brandt, V.~Conitzer, U.~Endriss, J.~Lang, and A.~D. Procaccia.
\newblock Introduction to computational social choice.
\newblock In F.~Brandt, V.~Conitzer, U.~Endriss, J.~Lang, and A.~D. Procaccia, editors, \emph{Handbook of Computational Social Choice}, chapter~1. Cambridge University Press, 2016.

\bibitem[Brandt et~al.(2023)Brandt, Greger, Segal-Halevi, and Suksompong]{BGSS23b}
F.~Brandt, M.~Greger, E.~Segal-Halevi, and W.~Suksompong.
\newblock Balanced donor coordination.
\newblock In \emph{Proceedings of the 24th ACM Conference on Economics and Computation (ACM-EC)}, page 299, 2023.

\bibitem[Brandt et~al.(2024)Brandt, Greger, Segal-Halevi, and Suksompong]{BGSS24a}
F.~Brandt, M.~Greger, E.~Segal-Halevi, and W.~Suksompong.
\newblock Optimal budget aggregation with single-peaked preferences.
\newblock In \emph{Proceedings of the 25th ACM Conference on Economics and Computation (ACM-EC)}, page~49, 2024.

\bibitem[Brandt et~al.(2025{\natexlab{a}})Brandt, Greger, Segal-Halevi, and Suksompong]{BGSS23a}
F.~Brandt, M.~Greger, E.~Segal-Halevi, and W.~Suksompong.
\newblock Coordinating charitable donations with {L}eontief preferences.
\newblock \emph{Journal of Economic Theory}, 230:\penalty0 106096, 2025{\natexlab{a}}.

\bibitem[Brandt et~al.(2025{\natexlab{b}})Brandt, Greger, Segal-Halevi, and Suksompong]{balancedArXiv}
F.~Brandt, M.~Greger, E.~Segal-Halevi, and W.~Suksompong.
\newblock Coordinating charitable donations with {L}eontief preferences.
\newblock Technical report, https://arxiv.org/abs/2305.10286, 2025{\natexlab{b}}.

\bibitem[Brânzei and Sandomirskiy(2019)]{branzei2019algorithms}
S.~Brânzei and F.~Sandomirskiy.
\newblock Algorithms for competitive division of chores.
\newblock Technical report, https://arxiv.org/abs/1907.01766, 2019.

\bibitem[Cabannes(2004)]{Caba04a}
Y.~Cabannes.
\newblock Participatory budgeting: a significant contribution to participatory democracy.
\newblock \emph{Environment and Urbanization}, 16\penalty0 (1):\penalty0 27--46, 2004.

\bibitem[Caragiannis et~al.(2015)Caragiannis, Fanelli, Gravin, and Skopalik]{caragiannis2015approximate}
I.~Caragiannis, A.~Fanelli, N.~Gravin, and A.~Skopalik.
\newblock Approximate pure nash equilibria in weighted congestion games: Existence, efficient computation, and structure.
\newblock \emph{ACM Transactions on Economics and Computation}, 3\penalty0 (1):\penalty0 2:1--2:32, 2015.

\bibitem[Chaudhury et~al.(2020)Chaudhury, Garg, McGlaughlin, and Mehta]{chaudhury2020dividing}
B.~R. Chaudhury, J.~Garg, P.~McGlaughlin, and R.~Mehta.
\newblock Dividing bads is harder than dividing goods: On the complexity of fair and efficient division of chores.
\newblock Technical report, https://arxiv.org/abs/2008.00285, 2020.

\bibitem[Chaudhury et~al.(2021)Chaudhury, Garg, McGlaughlin, and Mehta]{chaudhury2021competitive}
B.~R. Chaudhury, J.~Garg, P.~McGlaughlin, and R.~Mehta.
\newblock Competitive allocation of a mixed manna.
\newblock In \emph{Proceedings of the 32nd ACM-SIAM Annual Symposium on Discrete Algorithms (SODA)}, pages 1405--1424, 2021.

\bibitem[Chen et~al.(2004)Chen, Deng, Sun, and Yao]{chen2004fisher}
N.~Chen, X.~Deng, X.~Sun, and A.~C.-C. Yao.
\newblock Fisher equilibrium price with a class of concave utility functions.
\newblock In \emph{Proceedings of the 12th Conference on Annual European Symposium (ESA)}, pages 169--179, 2004.

\bibitem[Chen and Deng(2006)]{chen2006settling}
X.~Chen and X.~Deng.
\newblock Settling the complexity of two-player {N}ash equilibrium.
\newblock In \emph{Proceedings of the 47th Annual IEEE Symposium on Foundations of Computer Science (FOCS)}, pages 261--272, 2006.

\bibitem[Chen and Teng(2009)]{chen2009spending}
X.~Chen and S.-H. Teng.
\newblock Spending is not easier than trading: {O}n the computational equivalence of {F}isher and {A}rrow-{D}ebreu equilibria.
\newblock In \emph{Proceedings of the 20th International Symposium on Algorithms and Computation (ISAAC)}, pages 647--656, 2009.

\bibitem[Chen et~al.(2006)Chen, Deng, and Teng]{chen2006computing}
X.~Chen, X.~Deng, and S.-h. Teng.
\newblock Computing {N}ash equilibria: Approximation and smoothed complexity.
\newblock In \emph{Proceedings of the 47th Annual IEEE Symposium on Foundations of Computer Science (FOCS)}, pages 603--612, 2006.

\bibitem[Chen et~al.(2009)Chen, Dai, Du, and Teng]{chen2009settling}
X.~Chen, D.~Dai, Y.~Du, and S.~Teng.
\newblock Settling the complexity of {A}rrow-{D}ebreu equilibria in markets with additively separable utilities.
\newblock In \emph{Proceedings of the 50th Annual IEEE Symposium on Foundations of Computer Science (FOCS)}, pages 273--282, 2009.

\bibitem[Cheng et~al.(2017)Cheng, Diakonikolas, and Stewart]{cheng2017playing}
Y.~Cheng, I.~Diakonikolas, and A.~Stewart.
\newblock Playing anonymous games using simple strategies.
\newblock In \emph{Proceedings of the 28th Annual ACM-SIAM Symposium on Discrete Algorithms (SODA)}, pages 616--631, 2017.

\bibitem[Chien and Sinclair(2011)]{chien2011convergence}
S.~Chien and A.~Sinclair.
\newblock Convergence to approximate {N}ash equilibria in congestion games.
\newblock \emph{Games and Economic Behavior}, 71\penalty0 (2):\penalty0 315--327, 2011.

\bibitem[Codenotti and Varadarajan(2004)]{codenotti2004efficient}
B.~Codenotti and K.~Varadarajan.
\newblock Efficient computation of equilibrium prices for markets with {L}eontief utilities.
\newblock In \emph{International Colloquium on Automata, Languages, and Programming}, pages 371--382. Springer, 2004.

\bibitem[Codenotti et~al.(2004)Codenotti, Pemmaraju, and Varadarajan]{codenotti2004computation}
B.~Codenotti, S.~Pemmaraju, and K.~Varadarajan.
\newblock The computation of market equilibria.
\newblock \emph{SIGACT News}, 35\penalty0 (4):\penalty0 23--37, 2004.

\bibitem[Codenotti et~al.(2005{\natexlab{a}})Codenotti, McCune, Penumatcha, and Varadarajan]{codenotti2005market}
B.~Codenotti, B.~McCune, S.~Penumatcha, and K.~Varadarajan.
\newblock Market equilibrium for {CES} exchange economies: Existence, multiplicity, and computation.
\newblock In \emph{Proceedings of the 25th IARCS Annual Conference on Foundations of Software Technology and Theoretical Computer Science (FSTTCS)}, pages 505--516, 2005{\natexlab{a}}.

\bibitem[Codenotti et~al.(2005{\natexlab{b}})Codenotti, Pemmaraju, and Varadarajan]{codenotti2005polynomial}
B.~Codenotti, S.~Pemmaraju, and K.~Varadarajan.
\newblock On the polynomial time computation of equilibria for certain exchange economies.
\newblock In \emph{Proceedings of the 16th Annual ACM-SIAM Symposium on Discrete Algorithms (SODA)}, pages 72--81, 2005{\natexlab{b}}.

\bibitem[Codenotti et~al.(2006)Codenotti, Leoncini, and Resta]{codenotti2006efficient}
B.~Codenotti, M.~Leoncini, and G.~Resta.
\newblock Efficient computation of {N}ash equilibria for very sparse win-lose bimatrix games.
\newblock In \emph{Proceedings of the 14th Conference on Annual European Symposium (ESA)}, pages 232--243, 2006.

\bibitem[Cominetti et~al.(2009)Cominetti, Correa, and Stier-Moses]{cominetti2009impact}
R.~Cominetti, J.~R. Correa, and N.~E. Stier-Moses.
\newblock The impact of oligopolistic competition in networks.
\newblock \emph{Operations Research}, 57\penalty0 (6):\penalty0 1421--1437, 2009.

\bibitem[Cornes and Harley(2012)]{cornes2012fully}
R.~Cornes and R.~Harley.
\newblock Fully aggregative games.
\newblock \emph{Economics Letters}, 116:\penalty0 631--633, 2012.

\bibitem[Cummings et~al.(2015)Cummings, Kearns, Roth, and Wu]{cummings2015privacy}
R.~Cummings, M.~Kearns, A.~Roth, and Z.~S. Wu.
\newblock Privacy and truthful equilibrium selection for aggregative games.
\newblock In \emph{Proceedings of the 11th International Conference on Web and Internet Economics (WINE)}, pages 286--299, 2015.

\bibitem[Curtis~Eaves(1985)]{curtis1985finite}
B.~Curtis~Eaves.
\newblock Finite solution of pure trade markets with {C}obb-{D}ouglas utilities.
\newblock \emph{Economic Equilibrium: Model Formulation and Solution}, pages 226--239, 1985.

\bibitem[Curtis~Eaves(2009)]{curtis2009finite}
B.~Curtis~Eaves.
\newblock Finite solution of pure trade markets with {C}obb-{D}ouglas utilities.
\newblock In \emph{Economic Equilibrium: Model Formulation and Solution}, pages 226--239. Springer, 2009.

\bibitem[Daskalakis and Papadimitriou(2007)]{daskalakis2007computing}
C.~Daskalakis and C.~Papadimitriou.
\newblock Computing equilibria in anonymous games.
\newblock In \emph{Proceedings of the 48th Annual IEEE Symposium on Foundations of Computer Science (FOCS)}, pages 83--93, 2007.

\bibitem[Daskalakis and Papadimitriou(2015)]{daskalakis2015approximate}
C.~Daskalakis and C.~H. Papadimitriou.
\newblock Approximate {N}ash equilibria in anonymous games.
\newblock \emph{Journal of Economic Theory}, 156:\penalty0 207--245, 2015.

\bibitem[{de Berg} et~al.(2024){de Berg}, Freeman, {Schmidt-Kraepelin}, and Utke]{BFS+24a}
M.~{de Berg}, R.~Freeman, U.~{Schmidt-Kraepelin}, and M.~Utke.
\newblock Truthful budget aggregation: Beyond moving-phantom mechanisms.
\newblock Technical report, https://arxiv.org/pdf/2405.20303, 2024.

\bibitem[Debreu(1952)]{Debr52a}
G.~Debreu.
\newblock A social equilibrium existence theorem.
\newblock \emph{Proceedings of the National Academy of Sciences (PNAS)}, 38\penalty0 (10):\penalty0 886--893, 1952.

\bibitem[Debreu(1959)]{Debr59b}
G.~Debreu.
\newblock Topological methods in cardinal utility theory.
\newblock Cowles Foundation Discussion Paper~76, Cowles Foundation at Yale University, 1959.

\bibitem[Deligkas et~al.(2017)Deligkas, Fearnley, Savani, and Spirakis]{deligkas2017computing}
A.~Deligkas, J.~Fearnley, R.~Savani, and P.~Spirakis.
\newblock Computing approximate {N}ash equilibria in polymatrix games.
\newblock \emph{Algorithmica}, 77\penalty0 (2):\penalty0 487--514, 2017.

\bibitem[Deng et~al.(2002)Deng, Papadimitriou, and Safra]{deng2002complexity}
X.~Deng, C.~Papadimitriou, and S.~Safra.
\newblock On the complexity of equilibria.
\newblock In \emph{Proceedings of the 34th Annual ACM Symposium on Theory of Computing (STOC)}, pages 67--71, 2002.

\bibitem[Devanur and Kannan(2008)]{devanur2008marketfixed}
N.~R. Devanur and R.~Kannan.
\newblock Market equilibria in polynomial time for fixed number of goods or agents.
\newblock In \emph{Proveedings of the 49th Annual IEEE Symposium on Foundations of Computer Science}, pages 45--53, 2008.

\bibitem[Devanur et~al.(2008)Devanur, Papadimitriou, Saberi, and Vazirani]{devanur2008marketdual}
N.~R. Devanur, C.~H. Papadimitriou, A.~Saberi, and V.~V. Vazirani.
\newblock Market equilibrium via a primal--dual algorithm for a convex program.
\newblock \emph{Journal of the ACM}, 55\penalty0 (5):\penalty0 22:1--22:18, 2008.

\bibitem[Dubey et~al.(2006)Dubey, Haimanko, and Zapechelnyuk]{dubey2006strategic}
P.~Dubey, O.~Haimanko, and A.~Zapechelnyuk.
\newblock Strategic complements and substitutes, and potential games.
\newblock \emph{Games and Economic Behavior}, 54\penalty0 (1):\penalty0 77--94, 2006.

\bibitem[Elkind et~al.(2006)Elkind, Goldberg, and Goldberg]{elkind2006nash}
E.~Elkind, L.~A. Goldberg, and P.~Goldberg.
\newblock Nash equilibria in graphical games on trees revisited.
\newblock In \emph{Proceedings of the 7th ACM Conference on Electronic Commerce (ACM-EC)}, pages 100--109, 2006.

\bibitem[Elkind et~al.(2026)Elkind, Greger, Lederer, Suksompong, and Teh]{EGL+26a}
E.~Elkind, M.~Greger, P.~Lederer, W.~Suksompong, and N.~Teh.
\newblock Settling the score: Portioning with cardinal preferences.
\newblock \emph{Artificial Intelligence}, 352:\penalty0 104487, 2026.

\bibitem[Etessami and Yannakakis(2010)]{etessami2010complexity}
K.~Etessami and M.~Yannakakis.
\newblock On the complexity of {N}ash equilibria and other fixed points.
\newblock \emph{SIAM Journal on Computing}, 39\penalty0 (6):\penalty0 2531--2597, 2010.

\bibitem[Fabrikant et~al.(2004)Fabrikant, Papadimitriou, and Talwar]{fabrikant2004complexity}
A.~Fabrikant, C.~Papadimitriou, and K.~Talwar.
\newblock The complexity of pure {N}ash equilibria.
\newblock In \emph{Proceedings of the 36th Annual ACM Symposium on Theory of Computing (STOC)}, pages 604--612, 2004.

\bibitem[Freeman et~al.(2019)Freeman, Pennock, Peters, and Vaughan]{FPP+19a}
R.~Freeman, D.~M. Pennock, D.~Peters, and J.~W. Vaughan.
\newblock Truthful aggregation of budget proposals.
\newblock In \emph{Proceedings of the 2019 ACM Conference on Economics and Computation (ACM-EC)}, pages 751--752, 2019.

\bibitem[Freeman et~al.(2021)Freeman, Pennock, Peters, and {Wortman Vaughan}]{FPPV21a}
R.~Freeman, D.~M. Pennock, D.~Peters, and J.~{Wortman Vaughan}.
\newblock Truthful aggregation of budget proposals.
\newblock \emph{Journal of Economic Theory}, 193:\penalty0 105234, 2021.

\bibitem[Garg and McGlaughlin(2020)]{garg2020computing}
J.~Garg and P.~McGlaughlin.
\newblock Computing competitive equilibria with mixed manna.
\newblock In \emph{Proceedings of the 19th International Conference on Autonomous Agents and Multiagent Systems (AAMAS)}, pages 420--428, 2020.

\bibitem[Garg et~al.(2017)Garg, Mehta, Vazirani, and Yazdanbod]{garg2017settling}
J.~Garg, R.~Mehta, V.~V. Vazirani, and S.~Yazdanbod.
\newblock Settling the complexity of {L}eontief and {PLC} exchange markets under exact and approximate equilibria.
\newblock In \emph{Proceedings of the 49th Annual ACM Symposium on Theory of Computing (STOC)}, pages 890--901, 2017.

\bibitem[Goel et~al.(2019)Goel, Krishnaswamy, Sakshuwong, and Aitamurto]{GKSA19a}
A.~Goel, A.~K. Krishnaswamy, S.~Sakshuwong, and T.~Aitamurto.
\newblock Knapsack voting for participatory budgeting.
\newblock \emph{ACM Transactions on Economics and Computation}, 7\penalty0 (2):\penalty0 8:1--8:27, 2019.

\bibitem[Haret et~al.(2024)Haret, Klumper, Maly, and Sch{\"a}fer]{HKMS24a}
A.~Haret, S.~Klumper, J.~Maly, and G.~Sch{\"a}fer.
\newblock Committees and equilibria: Multiwinner approval voting through the lens of budgeting games.
\newblock In \emph{Proceedings of the 25th ACM Conference on Economics and Computation (ACM-EC)}, pages 51--70, 2024.

\bibitem[Harks and Timmermans(2022)]{harks2022equilibrium}
T.~Harks and V.~Timmermans.
\newblock Equilibrium computation in resource allocation games.
\newblock \emph{Mathematical Programming}, 194\penalty0 (1):\penalty0 1--34, 2022.

\bibitem[Huang(2013)]{huang2013collusion}
C.-C. Huang.
\newblock Collusion in atomic splittable routing games.
\newblock \emph{Theory of Computing Systems}, 52\penalty0 (4):\penalty0 763--801, 2013.

\bibitem[Jain(2007)]{Jain07a}
K.~Jain.
\newblock A polynomial time algorithm for computing an {A}rrow--{D}ebreu market equilibrium for linear utilities.
\newblock \emph{SIAM Journal on Computing}, 37\penalty0 (1):\penalty0 303--318, 2007.

\bibitem[Janson(2016)]{Jans16a}
S.~Janson.
\newblock {P}hragm{\'e}n's and {T}hiele's election methods.
\newblock Technical Report arXiv:1611.08826 [math.HO], arXiv.org, 2016.

\bibitem[Jensen(2010)]{jensen2010aggregative}
M.~Jensen.
\newblock Aggregative games and best-reply potentials.
\newblock \emph{Economic Theory}, 43:\penalty0 45--66, 2010.

\bibitem[Jensen(2018)]{jensen2018aggregative}
M.~Jensen.
\newblock Aggregative games.
\newblock In \emph{Handbook of Game Theory and Industrial Organization, Volume I}, pages 66--92. Edward Elgar Publishing, 2018.

\bibitem[Kearns and Mansour(2002)]{kearns2002efficient}
M.~Kearns and Y.~Mansour.
\newblock Efficient {N}ash computation in large population games with bounded influence.
\newblock In \emph{Proceedings of the 18th Conference on Uncertainty in Artificial Intelligence}, pages 259--266, 2002.

\bibitem[Kearns et~al.(2001)Kearns, Littman, and Singh]{kearns2001graphical}
M.~Kearns, M.~L. Littman, and S.~Singh.
\newblock Graphical models for game theory.
\newblock \emph{Proceedings of the 17th Conference on Uncertainty in Artificial Intelligence}, pages 253--260, 2001.

\bibitem[Klimm and Warode(2025)]{klimm2025complexity}
M.~Klimm and P.~Warode.
\newblock Complexity and parametric computation of equilibria in atomic splittable congestion games via weighted block laplacians.
\newblock \emph{SIAM Journal on Computing}, 54\penalty0 (5):\penalty0 1241--1293, 2025.

\bibitem[Kukushkin(1994)]{kukushkin1994fixed}
N.~S. Kukushkin.
\newblock A fixed-point theorem for decreasing mappings.
\newblock \emph{Economics Letters}, 46\penalty0 (1):\penalty0 23--26, 1994.

\bibitem[Lackner and Skowron(2023)]{LaSk23a}
M.~Lackner and P.~Skowron.
\newblock \emph{Multi-winner voting with approval preferences}.
\newblock Springer Nature, 2023.

\bibitem[Lemke and Howson(1964)]{lemke1964equilibrium}
C.~E. Lemke and J.~T. Howson.
\newblock Equilibrium points of bimatrix games.
\newblock \emph{SIAM Journal on Applied Mathematics}, 12\penalty0 (2):\penalty0 413--423, 1964.

\bibitem[Li et~al.(2025)Li, Li, Ran, Zheng, and Huang]{li2025survey}
H.~Li, J.~Li, L.~Ran, L.~Zheng, and T.~Huang.
\newblock A survey of distributed algorithms for aggregative games.
\newblock \emph{IEEE/CAA Journal of Automatica Sinica}, 12\penalty0 (5):\penalty0 859--871, 2025.

\bibitem[Lindner(2011)]{Lind11a}
T.~Lindner.
\newblock \emph{Zur Manipulierbarkeit der Allokation {\"o}ffentlicher G{\"u}ter: Theoretische Analyse und Simulationsergebnisse}.
\newblock PhD thesis, Karlsruhe Institute of Technology, 2011.

\bibitem[Lipton et~al.(2003)Lipton, Markakis, and Mehta]{lipton2003playing}
R.~J. Lipton, E.~Markakis, and A.~Mehta.
\newblock Playing large games using simple strategies.
\newblock In \emph{Proceedings of the 4th ACM conference on Electronic Commerce}, pages 36--41, 2003.

\bibitem[McKenzie(1959)]{McKe59a}
L.~W. McKenzie.
\newblock On the existence of general equilibrium for a competitive market.
\newblock \emph{Econometrica}, 27\penalty0 (1):\penalty0 54--71, 1959.

\bibitem[Merrill(1972)]{merrill1972applications}
O.~H. Merrill.
\newblock \emph{Applications and Extensions of an algorithm that computes fixed points of certain upper semi-continuous point to set mappings}.
\newblock PhD thesis, University of Michigan, 1972.

\bibitem[Nash(1950)]{Nash50a}
J.~F. Nash.
\newblock Equilibrium points in $n$-person games.
\newblock \emph{Proceedings of the National Academy of Sciences (PNAS)}, 36:\penalty0 48--49, 1950.

\bibitem[Nash(1951)]{Nash51a}
J.~F. Nash.
\newblock Non-cooperative games.
\newblock \emph{Annals of Mathematics}, 54\penalty0 (2):\penalty0 286--295, 1951.

\bibitem[Orda et~al.(2002)Orda, Rom, and Shimkin]{orda2002competitive}
A.~Orda, R.~Rom, and N.~Shimkin.
\newblock Competitive routing in multiuser communication networks.
\newblock \emph{IEEE/ACM Transactions on networking}, 1\penalty0 (5):\penalty0 510--521, 2002.

\bibitem[Orlin(2010)]{orlin2010improved}
J.~B. Orlin.
\newblock Improved algorithms for computing {F}isher's market clearing prices.
\newblock In \emph{Proceedings of the 42nd Annual ACM Symposium on Theory of Computing (STOC)}, pages 291--300, 2010.

\bibitem[Orzech and Rinard(2025)]{orzech2025nash}
E.~Orzech and M.~Rinard.
\newblock {N}ash equilibria with irradical probabilities.
\newblock Technical report, https://arxiv.org/abs/2507.09422, 2025.

\bibitem[Papadimitriou et~al.(2023)Papadimitriou, Vlatakis-Gkaragkounis, and Zampetakis]{papadimitriou2023computational}
C.~Papadimitriou, E.-V. Vlatakis-Gkaragkounis, and M.~Zampetakis.
\newblock The computational complexity of multi-player concave games and {K}akutani fixed points.
\newblock In \emph{Proceedings of the 24th ACM Conference on Economics and Computation}, pages 1045--1045, 2023.

\bibitem[Papadimitriou(1994)]{Papa94b}
C.~H. Papadimitriou.
\newblock On the complexity of the parity argument and other inefficient proofs of existence.
\newblock \emph{Journal of Computer and System Sciences}, 48\penalty0 (3):\penalty0 498--532, 1994.

\bibitem[Parise et~al.(2020)Parise, Grammatico, Gentile, and Lygeros]{parise2020distributed}
F.~Parise, S.~Grammatico, B.~Gentile, and J.~Lygeros.
\newblock Distributed convergence to nash equilibria in network and average aggregative games.
\newblock \emph{Automatica}, 117:\penalty0 108959, 2020.

\bibitem[Peters and Skowron(2020)]{PeSk20a}
D.~Peters and P.~Skowron.
\newblock Proportionality and the limits of welfarism.
\newblock In \emph{Proceedings of the 21nd ACM Conference on Economics and Computation (ACM-EC)}, pages 793--794, 2020.

\bibitem[Peters et~al.(2021)Peters, Pierczy\'nski, and Skowron]{PPS21a}
D.~Peters, G.~Pierczy\'nski, and P.~Skowron.
\newblock Proportional participatory budgeting with additive utilities.
\newblock In \emph{Proceedings of the 35th Annual Conference on Neural Information Processing Systems (NeurIPS)}, pages 12726--12737, 2021.

\bibitem[Phragm{\'e}n(1899)]{Phra99a}
E.~Phragm{\'e}n.
\newblock Till fr{\aa}gan om en proportionell valmetod.
\newblock \emph{Statsvetenskaplig Tidskrift}, 2\penalty0 (2):\penalty0 297--305, 1899.

\bibitem[Porter et~al.(2008)Porter, Nudelman, and Shoham]{porter2008simple}
R.~Porter, E.~Nudelman, and Y.~Shoham.
\newblock Simple search methods for finding a {N}ash equilibrium.
\newblock \emph{Games and Economic Behavior}, 63\penalty0 (2):\penalty0 642--662, 2008.

\bibitem[Rosen(1965)]{rosen1965existence}
J.~B. Rosen.
\newblock Existence and uniqueness of equilibrium points for concave n-person games.
\newblock \emph{Econometrica}, 33\penalty0 (3):\penalty0 520--534, 1965.

\bibitem[Roughgarden and Schoppmann(2015)]{roughgarden2015local}
T.~Roughgarden and F.~Schoppmann.
\newblock Local smoothness and the price of anarchy in splittable congestion games.
\newblock \emph{Journal of Economic Theory}, 156:\penalty0 317--342, 2015.

\bibitem[Rubinstein(2015)]{rubinstein2015inapproximability}
A.~Rubinstein.
\newblock Inapproximability of {N}ash equilibrium.
\newblock In \emph{Proceedings of the 47th Annual ACM Symposium on Theory of Computing}, pages 409--418, 2015.

\bibitem[Scarf(1967)]{scarf1967computation}
H.~E. Scarf.
\newblock On the computation of equilibrium prices.
\newblock Technical Report 232, Cowles Foundation Discussion Papers, 1967.

\bibitem[Tsaknakis and Spirakis(2007)]{tsaknakis2007optimization}
H.~Tsaknakis and P.~G. Spirakis.
\newblock An optimization approach for approximate {N}ash equilibria.
\newblock In \emph{Internet and Network Economics}, pages 42--56. Springer, 2007.

\bibitem[Vazirani(2007)]{vazirani2007combinatorial}
V.~V. Vazirani.
\newblock Combinatorial algorithms for market equilibria.
\newblock In \emph{Algorithmic Game Theory}, chapter~5. Cambridge University Press, 2007.

\bibitem[von Neumann(1928)]{vonneumann1928theorie}
J.~von Neumann.
\newblock Zur {T}heorie der {G}esellschaftsspiele.
\newblock \emph{Mathematische Annalen}, 100\penalty0 (1):\penalty0 295--320, 1928.

\bibitem[von Neumann and Morgenstern(1944)]{vNM44a}
J.~von Neumann and O.~Morgenstern.
\newblock \emph{Theory of Games and Economic Behavior}.
\newblock Princeton University Press, 1944.

\bibitem[Wang and Nedić(2024)]{wang2024differentially}
Y.~Wang and A.~Nedić.
\newblock Differentially private distributed algorithms for aggregative games with guaranteed convergence.
\newblock \emph{IEEE Transactions on Automatic Control}, 69\penalty0 (8):\penalty0 5168--5183, 2024.

\bibitem[Ye(2008)]{ye2008path}
Y.~Ye.
\newblock A path to the {A}rrow--{D}ebreu competitive market equilibrium.
\newblock \emph{Mathematical Programming}, 111\penalty0 (1):\penalty0 315--348, 2008.

\end{thebibliography}

\end{document}